\documentclass[11pt]{article}

\usepackage[T1]{fontenc}
\usepackage[utf8]{inputenc}
\usepackage{lmodern}
\usepackage{microtype}
\usepackage[a4paper,margin=27mm]{geometry}
\usepackage{amsmath,amssymb,amsthm,mathtools,bm}
\usepackage{booktabs,longtable,array}
\usepackage{float}
\usepackage{enumitem}
\usepackage[hidelinks]{hyperref}
\usepackage{url}
\usepackage{xcolor}
\usepackage[normalem]{ulem}
\usepackage{graphicx}
\usepackage{comment}

\usepackage{authblk}

\allowdisplaybreaks
\newtheorem{theorem}{Theorem}[section]
\newtheorem{proposition}[theorem]{Proposition}
\newtheorem{lemma}[theorem]{Lemma}
\newtheorem{corollary}[theorem]{Corollary}
\theoremstyle{definition}

\theoremstyle{remark}
\newtheorem{remark}[theorem]{Remark}

\newcommand{\A}{\mathcal A}

\newcommand{\Pp}{\mathcal P}

\newcommand{\Perm}{\mathcal S}

\newcommand{\LP}{\mathrm{LP}}
\newcommand{\BHL}{\mathrm{BHL}}

\newcommand{\CW}{\mathrm{cw}}

\newcommand{\BE}{\mathsf{BE}}

\DeclareMathOperator{\comp}{comp}
\DeclareMathOperator{\wt}{wt}
\DeclareMathOperator{\diam}{diam}

\title{Envelopes of upper bounds for nonbinary constant-weight and constant-composition codes}

\author[1]{Artur Akhiiarov}
\author[2]{Peter Boyvalenkov}
\author[2]{Danila Cherkashin}
\author[3, 4, 5]{Andrei Raigorodskii}

\affil[1]{Laboratory of Advanced
Combinatorics and Network Applications,
Moscow Institute of Physics and Technology (State University),
Moscow, Russia}
\affil[2]{Institute of Mathematics and Informatics, Bulgarian Academy of Sciences, Sofia}
\affil[3]{Department of Discrete Mathematics and Laboratory of Advanced
Combinatorics and Network Applications,
Moscow Institute of Physics and Technology (State University),
Moscow, Russia}

\affil[4]{Department of Mathematical Statistics and Random Processes,
Faculty of Mechanics and Mathematics, Lomonosov Moscow State University,
Moscow, Russia}

\affil[5]{Caucasus Mathematical Center, Adyghe State University,
Maykop, Republic of Adygea, Russia}

\date{September 2026}

\begin{document}
\maketitle

\begin{abstract}
Deriving upper bounds on code size from existing bounds is a classical approach in coding theory, dating back to the seminal results of Elias, Bassalygo, 
and Levenshtein. We study a framework encompassing the Bassalygo--Elias and Levenshtein inequalities for binary and nonbinary (constant-weight) codes and provides certain generalizations. The asymptotic cost of transferring a bound between different symbol compositions is expressed in terms of mutual information, 
yielding an information-theoretic optimal transport formulation.

We determine the optimal permutation-transport cost between arbitrary
compositions in terms of their least common majorant in the majorization
order. Specializing to symmetric constant-weight compositions yields
explicit transport profiles. As a byproduct, we establish unimodality
of the asymptotic constant-weight rate as a function of the relative weight. 
We prove that the resulting closure operators are idempotent
and that applying transport before outer Bassalygo--Elias averaging
leaves the unrestricted bound obtained from the same input unchanged.
We also establish necessary and sufficient conditions for an upper bound to be a fixed point of the closure operator.
Since the asymptotic rate function is an upper bound for itself and it is a fixed point, we conclude the Schur concavity of the constant-composition rate function.

Finally, we survey existing upper bounds for binary and nonbinary
constant-weight and constant-composition codes, combine them into
optimized envelopes within the transport framework, and obtain
improved theoretical and numerical bounds.
\end{abstract}

\section{Introduction}

Fix an alphabet $\A$ with $|\A|=q$, $q \geq 2$ is an integer, and a distinguished symbol $0$. 
A \textit{code} $C$ of length $n$ is any non-empty subset of $\A^n$. A \textit{word} $x$ is an element of $C$. The weight of a word $x$ is the number $\wt(x)$ of non-zero symbols in $x$.
A constant-weight code with weight $w$ is a code in which every word has weight $w$. Hamming distance $d_H(x,y)$ between words $x$ and $y$ is the number of positions in which $x$ and $y$ differ. The minimum distance $d(C)$ of a code $C$ is the smallest value of $d_H(x,y)$ over distinct $x,y \in C$.

Two classical problems of coding theory consider the maximum possible cardinality of codes with fixed length and minimum distance and fixed length, constant weight, and minimum distance, respectively. Thus, we are interested in asymptotic upper bounds for the quantities
\[ 
A_q(n,d):= \max \left\{ |C|:C\subseteq\A^n, \ d(C) \geq d \right\}, 
\]
\[ 
A_q(n,d,w) := \max\left\{|C|:C\subseteq\A^n,\ \wt(x)=w \text{ for all } x \in C, \ d(C) \geq d \right\}. 
\]
We refer to the books~\cite[Chapter 17]{MacWilliams-Sloane},~\cite[Chapters 5--7]{Coding-Theory-Handbook},~\cite[Chapters 1.9, 12, 13]{Conz-Encycl},~\cite[Chapter 3]{Tsfasman-Vlăduţ-AG_codes}, and references therein from the enormous body of literature on the subject. 

When considering codes as the length $n$ tends to infinity, one analyzes, respectively, the tradeoff between the relative minimum distance $d(C)/n \to \delta$ for $q$-ary codes and asymptotic behavior of
\[
 R_q^{\mathrm{unr}}(\delta) := \limsup_{n\to\infty} \frac{1}{n}\log_q A_q(n,\lceil\delta n\rceil).
\]
and between the relative weight $w/n \to \omega$, the relative distance $d(C)/n \to \delta$, and the asymptotic behavior of 
\[ 
R_q(\delta,\omega):=\limsup_{n\to\infty}\frac{1}{n}\log_q A_q(n,\lceil\delta n\rceil, \lceil \omega n \rceil) 
\]
for constant-weight codes (everywhere, $\delta,\omega \in [0,1]$ are real parameters). 

We first derive a broad generalization of Levenshtein's bound~\cite{Levenshtein1971} from 1971. 
We prove two inequalities (Theorems~\ref{thm:finite-matrix-transport} and~\ref{thm:finite-permutation-transport}) to provide transport between bounds for constant composition codes of the same distance. 
Setting a general framework, we derive an information-theoretic optimal transport formulation in which we find exact constant-composition transport cost (Theorem~\ref{thm:cc-join-cost}) and propose a plan for its ``realization''. 
In the basic case, when both source and target are constant-weight, the problem is reduced to explicit computations leading to certain bounds (Corollary~\ref{cor:explicit-profile}). Furthermore, the general framework allows us to define and investigate a closure operator which can be applied to all valid asymptotic upper bounds. We prove and study two major features of this operator, its idempotence (Theorem~\ref{thm:idempotence}; i.e., one application is enough) and its majorization property (Proposition~\ref{prop:uniform-composition}; explicitly showing that the constant-weight case provides optimality over the composition vectors). The majorization monotonicity gives, in particular, $q$-ary unimodality of the asymptotic constant-composition rate which was previously considered only in the binary case~\cite{Samorodnitsky2001,Cohen2010,BachocEtAl2011}. 
The properties of an upper bound that cannot be improved at any source by our transport approach are investigated for obtaining necessary and sufficient conditions for unimprovability (Theorems~\ref{thm:fp-composition} and~\ref{thm:fp-weight}). 
One of the main theoretical contributions of the paper (Corollary~\ref{cor:fp-exact-rates}) shows that the constant-composition rate is Schur-concave.

Our general framework also allows unified treatment of all pre-existing bounds. This allows us to provide theoretical and numerical comparisons. Since not all known estimates have a simple description, a full explicit theoretical envelope looks difficult.
Instead, some theoretical properties of the full envelope are provided in Section~\ref{sec:master-envelope}. 
They are based on the numerical observations which are discussed in Section~\ref{sec:numerics}.

Let us note that an immediate application of Levenshtein theorem from 1971 (inequality~\eqref{eq:levenshtein-specialization} below) to OpenAI moving subspace constant weight binary bound (see~\cite[Chapter 2, Lemma 3.5]{OpenAI2026}) already gives some improvements. For instance, the upper bound on $R_2(0.1,0.2)$ changes from $0.42040\dots$ to $0.4145$ (see also Fig.~\ref{fig:binary-cw-comparison}). We provide examples (tables and graphs) of our refinements compared with best-known upper bounds for some values of the alphabet size $q$.

The paper is organized as follows. In Section~\ref{sec:mainTool} we describe our main inequalities connecting upper bound for different regimes. Section~\ref{sec:anal} and~\ref{sec:closure} are devoted to theoretical analysis of these inequalities in the asymptotic regime. 
Then Section~\ref{sec:base-bounds} collects known upper bounds which we combine with our approach theoretically in Section~\ref{sec:master-envelope} and numerically in Section~\ref{sec:numerics}. Appendix A contains a relation of our optimal transport with the classical rate--distortion problem.

\paragraph{AI usage.} The authors used OpenAI ChatGPT Pro to refine the proofs,
check calculations, and assist with manuscript preparation. All mathematical
statements, proofs, computations, and citations were independently verified by the
authors, who take full responsibility for the content of the manuscript.

\paragraph{Acknowledgments.} DC is supported by the Bulgarian National Science Fund Contract KP-06-N72/6-2023.
PB is supported by the Ministry of Education and Science of Bulgaria under the Grant No. DO1-98/26.06.2025 ``National Centre for High Performance and Distributed Computing'' (NCHPDC). The authors would like to thank Alexander Barg for providing insightful remarks. 

\section{Main tools}
\label{sec:mainTool}

We identify the alphabet $\A$ with the cyclic group $\mathbb Z_q$.
The composition vector of a word $c$ is the vector $(c_0, \dots, c_{q-1})$, where $c_i$ stands for the number of symbols $i$ in $c$.
A code $C$ is called constant-composition if all $x \in C$ have the same composition vector. If $C$ is constant-composition with 
a composition vector $\bm w$, we write $\comp(C)=\bm w$. 
For an integer composition vector
$\bm w=(w_a)_{a\in\A}$ with $\sum_a w_a=n$, let 
\[
 A_q(n,d;\bm w) := \max\bigl\{|C|: C\subseteq\A^n,\ 
 \comp(C)=\bm w,\
 d_H(x,y)\geq d\text{ for }x\neq y
 \bigr\}.
\]
The quantity $A_q(n,d;\bm w)$ is more general than the classical $A_q(n,d,w)$ and correspondingly allows more general treatment of the asymptotic bounds problem.

For a rational probability vector $p=(p_a)_{a\in\A}$ and admissible
integer sequences, write
\[
 R_q(\delta;p)
 :=
 \limsup_{n\to\infty}
 \frac{1}{n}\log_q A_q(n,\lceil\delta n\rceil;np).
\]
For the symmetric composition put $p_\omega := \left(1-\omega,
 \frac{\omega}{q-1},
 \ldots,
 \frac{\omega}{q-1}
 \right)$; Proposition~\ref{prop:uniform-composition} below shows that $R_q(\delta,\omega) = R_q(\delta;p_\omega)$ so the asymptotics of a constant-weight rate function coincides with the asymptotics of the corresponding symmetric composition.

Let $\tau=(\tau_{ab})_{a,b\in\mathbb Z_q}$ be a rational nonnegative $q\times q$ matrix of total mass one.
We regard $\tau_{ab}$ as the proportion of coordinates on which the source symbol $a$ is sent to the target symbol $b$. Denote by
\begin{equation}
 p_a=\sum_{b\in\mathbb Z_q}\tau_{ab},
 \qquad
 r_b=\sum_{a\in\mathbb Z_q}\tau_{ab},
 \label{eq:transport-marginals}
\end{equation}
the row and column sums in the matrix $\tau$, respectively. 
Thus $p=(p_a)_{a \in \mathbb{Z}_q}$ and $r=(r_b)_{b \in \mathbb{Z}_q}$ are the source and target compositions, respectively.

For $u\in\mathbb Z_q$, define the distribution of the additive shift
\begin{equation} \label{eq:shift-marginal}
 \mu_u := \sum_{a\in\mathbb Z_q}\tau_{a,a+u},
\end{equation}
where the second index is taken modulo $q$.

Our first main tool is the following theorem. Recall the definition of a multinomial coefficient 
\[
\binom{n}{v} := \binom{n}{v_1, \dots, v_k} = \frac{n!}{v_1! \cdot \dots \cdot v_k!}, 
\]
where $v$ stands for the $k$-tuple $(v_1,\ldots,v_k)$ with $v_1 + \dots + v_k = n$.

\begin{theorem}[Transition-matrix transport inequality]
\label{thm:finite-matrix-transport}
Let $p$ and $r$ be probability vectors and $\tau$ be a matrix with non-negative entries satisfying~\eqref{eq:transport-marginals}; assume that all $n\tau_{ab}$ are integers. 
Then, for every $n$ and $d$,
\begin{equation} \label{eq:finite-matrix-transport}
 A_q(n,d;np) 
 \leq
 \frac{
 \displaystyle
 \binom{n}{(n\mu_u)_{u\in\mathbb Z_q}}
 }{
 \displaystyle
 \prod_{a\in\mathbb Z_q}
 \binom{np_a}{(n\tau_{ab})_{b\in\mathbb Z_q}}
 }
 A_q(n,d;nr).
\end{equation}
\end{theorem}

\begin{proof}
Let $C \subset \mathbb{Z}_q^n$ be a code of composition $np$ and size $A_q(n,d;np)$.
For each $x=(x_1,\ldots,x_n)\in C$, consider all vectors
$u=(u_1,\ldots,u_n)\in\mathbb Z_q^n$,
such that
\[
\left| \{i:x_i=a,\ x_i+u_i=b\} \right | = n \tau_{ab} \ \
\mbox{for every} \ a,b\in\mathbb Z_q.
\]
For a fixed $x$, on the $np_a$ coordinates carrying the symbol $a$,
we choose which $n\tau_{ab}$ coordinates are sent to each
$b\in\mathbb Z_q$. Hence the number of admissible vectors $u$ is given by the product of multinomial coefficients 
\[\prod_{a\in\mathbb Z_q}
 \binom{np_a}{(n\tau_{ab})_{b\in\mathbb Z_q}}. \]

If a coordinate undergoes the transition $a\mapsto b$, then necessarily $u_i=b-a\pmod q$.

Therefore every admissible $u$ has composition $n\mu$, where $\mu$ is
defined by~\eqref{eq:shift-marginal}. The total number of vectors in
$\mathbb Z_q^n$ with this composition is
\[
 \binom{n}{(n\mu_u)_{u\in\mathbb Z_q}}.
\]

Now fix such a vector $u$. The map
\[
 T_u:\mathbb Z_q^n\longrightarrow\mathbb Z_q^n,
 \qquad
 T_u(x)=x+u
\]
is a Hamming isometry. Moreover, for every pair $(x,u)$ counted above,
the image $T_u(x)$ has composition $nr$. Consequently, for a fixed
$u$, at most $A_q(n,d;nr)$ codewords of $C$ can occur in admissible pairs $(x,u)$.

Double counting the admissible pairs $(x,u)$ therefore gives
\[|C|
 \prod_{a\in\mathbb Z_q}
 \binom{np_a}{(n\tau_{ab})_{b\in\mathbb Z_q}}
 \leq
 \binom{n}{(n\mu_u)_{u\in\mathbb Z_q}}
 A_q(n,d;nr), \]
which is exactly~\eqref{eq:finite-matrix-transport}, completing the proof.
\end{proof}

Theorem~\ref{thm:finite-matrix-transport} uses only additive shifts in $\A$. 
We prove the following extension (Theorem~\ref{thm:finite-permutation-transport}) to arbitrary
coordinatewise permutations of $\A$. For the symmetric source and
target compositions analyzed below, the optimal information cost is
already attained by additive shifts; see Proposition~\ref{prop:additiveShifts}.
The permutation extension is nevertheless essential for our proofs of majorization monotonicity and for the reduction of the fixed-weight problem to the balanced
constant-composition problem, see Corollary~\ref{cor:majorization} and Proposition~\ref{prop:uniform-composition}.

Let $\Perm_q$ be the symmetric group on $\A$. Let
$\eta=(\eta_{a,\pi})_{a\in\A,\pi\in\Perm_q}$ be a rational probability
distribution ($\eta$ is a counterpart of $\tau$). Write
\begin{equation} \label{eq:perm-trans}
 p_a=\sum_{\pi}\eta_{a,\pi},\qquad
 \mu_\pi=\sum_a\eta_{a,\pi},\qquad
 r_b=\sum_{\substack{a,\pi\\\pi(a)=b}}\eta_{a,\pi}.
 \end{equation}

\begin{theorem}[Permutation-transport inequality]
\label{thm:finite-permutation-transport}
Let $p$ and $r$ be probability vectors and $\eta$ be a rational probability
distribution satisfying~\eqref{eq:perm-trans}; assume that all $n\eta_{a\pi}$ are integers. 
\begin{equation}
 A_q(n,d;np)
 \leq
 \frac{\displaystyle\binom{n}{(n\mu_\pi)_{\pi\in\Perm_q}}}
 {\displaystyle\prod_{a\in\A}
 \binom{np_a}{(n\eta_{a,\pi})_{\pi\in\Perm_q}}}
 A_q(n,d;nr).
 \label{eq:finite-permutation-transport}
\end{equation}
\end{theorem}

\begin{proof}
Let $C \subset \A^n$ be a code with $\comp(C) = np$ and $|C| = A_q(n,d;np)$. For each $x\in C$, count all sequences
$\bm\pi=(\pi_1,\ldots,\pi_n)\in\Perm_q^n$ such that
\[
 \left |\{i:x_i=a,\ \pi_i=\pi\} \right| = n\eta_{a,\pi}.
\]
For fixed $x$ their number is the denominator in the ratio in
\eqref{eq:finite-permutation-transport}. The total number of permutation
sequences with marginal type $\mu = (\mu_\pi)_{\pi \in \Perm_q}$ is the numerator. For a fixed
$\bm\pi$, the map
\[
 x\longmapsto \bm\pi(x):=(\pi_1(x_1),\ldots,\pi_n(x_n))
\]
is a Hamming isometry, and every selected image has composition $nr$.
Consequently at most $A_q(n,d;nr)$ codewords can be paired with a fixed
$\bm\pi$. Double counting the pairs $(x,\bm\pi)$ proves the claim.
\end{proof}

Restricting $\pi_i$ to the $q$ translations of a cyclic group recovers Theorem~\ref{thm:finite-matrix-transport}.

\subsection{Levenshtein's theorem revisited}

Theorem~\ref{thm:finite-matrix-transport} should be viewed as a
nonbinary constant-composition generalization of the double-counting
inequality used by Levenshtein~\cite{Levenshtein1971} for binary
constant-weight codes:
\begin{equation} \label{eq:levenshtein-specialization}
 A_2(n,d,w) \leq \frac{\binom{n}{w-\beta+\alpha}}{\binom{n-w}{\alpha}\binom{w}{\beta}} A_2(n,d,\alpha+\beta),
\end{equation}
where $\alpha\leq n-w$ and $\beta\leq w$.

Indeed, for $q=2$ take the transition matrix
\[
 n\tau
 =
 \begin{pmatrix}
 n-w-\alpha & \alpha\\
 w-\beta & \beta
 \end{pmatrix}.
\]
Thus $\alpha$ zero coordinates are changed from $0$ to $1$,
$w-\beta$ one coordinates are changed from $1$ to $0$, and the
remaining coordinates are left unchanged. The target word therefore
has weight $\alpha+\beta$.
The associated binary shift vector has weight
$\alpha+(w-\beta)=w-\beta+\alpha$.
Hence Theorem~\ref{thm:finite-matrix-transport} gives exactly~\eqref{eq:levenshtein-specialization}.

\subsection{Bassalygo averaging} \label{bassalygo}

We recall the averaging step that passes from a constant-weight bound to an unrestricted bound. 

\begin{lemma}[Bassalygo averaging] \label{lem:bassalygo-averaging} 
For every $n,d,w$, 
\begin{equation} \label{eq:finite-bassalygo}
A_q(n,d) \leq \frac{q^n}{\binom nw(q-1)^w}\, A_q(n,d,w). 
\end{equation} 
\end{lemma} 

\begin{proof} 
Denote the shell of words of length $n$ and weight $w$ by $J_q(n,w)$, i.e. 
\begin{equation} \label{eq:defShell}
J_q(n,w):=\{x\in\A^n:\wt(x)=w\}, \qquad |J_q(n,w)|=\binom nw(q-1)^w. 
\end{equation}
For an unrestricted code $C\subseteq\A^n$ of minimum distance at least $d$ and cardinality $A_q(n,d)$, average $|(C+z)\cap J_q(n,w)|$ over $z\in\A^n$. Every pair $(c,s)\in C\times J_q(n,w)$ determines a unique translation $z=s-c$. Hence 
\[ 
\sum_{z\in\A^n}|(C+z)\cap J_q(n,w)| = |C| \cdot |J_q(n,w)|. 
\] 
For some $z$, 
\[ 
|(C+z)\cap J_q(n,w)| \geq |C|\frac{|J_q(n,w)|}{q^n}. 
\] 
Translation is a Hamming isometry, so this intersection is a constant-weight code of minimum distance at least $d$. 
Therefore 
\[ 
|C|\frac{|J_q(n,w)|}{q^n} \leq A_q(n,d,w),
\] 
which proves~\eqref{eq:finite-bassalygo}. 
\end{proof}

\section{Asymptotic transport inequalities from Theorem~\ref{thm:finite-permutation-transport}} \label{sec:anal}

In this section we analyze how to apply Theorem~\ref{thm:finite-permutation-transport} in the most effective way and justify the obtained asymptotic bounds.
It is convenient to use the language of information-theoretic transport, with mutual information as the cost.

First, note that there are generally several transport plans $\eta$ with the same source and target marginals $p$ and $r$.
We determine the optimal permutation-transport cost between arbitrary
compositions through their least common majorant, then specialize the
result to the constant-weight case. 

\subsection{Entropy and mutual information}
The definition of mutual information, the chain rule, and the data processing inequality are standard; see~\cite[Sections~2.3--2.8]{CoverThomas2006} and~\cite[Chapter~2]{CsiszarKorner2011}. All random variables considered here take values in finite sets,
and all logarithms are to base $q$.
For the historical origin of these notions, see~\cite{Shannon1948}.

For a probability distribution $P$ on a set $\mathcal S$,
define 
\[
 H_q(P):=-\sum_{x\in\mathcal S}P(x)\log_qP(x).
\]
For a random variable $X$ with distribution $P_X$, write $H_q(X):=H_q(P_X)$.
Then define the $q$-ary entropy function
\[
 h_q(x) := -x\log_q x -(1-x)\log_q(1-x) + x\log_q(q-1).
\]
We use the convention $0\log 0=0$.
Thus $H_q(p_\omega)=h_q(\omega)$, where $p_\omega$ is the constant-composition vector $\left(1-\omega,
 \frac{\omega}{q-1},
 \ldots,
 \frac{\omega}{q-1}
 \right)$. 

\paragraph{Conditional entropy and mutual information.} 
For two random variables $X,Y$, write
$H_q(X,Y):=H_q(P_{XY})$. Their conditional entropy is
\[
 H_q(X\mid Y)
 :=
 \sum_{y:\,P_Y(y)>0}
 P_Y(y)\,H_q(P_{X\mid Y=y}),
 \qquad
 P_{X\mid Y=y}(x):=\frac{P_{XY}(x,y)}{P_Y(y)}.
\]
Thus $H_q(X\mid Y)$ is the average entropy of the conditional
distributions. Conditioning on several variables means conditioning
on their tuple, for example
$H_q(X\mid Y,Z):=H_q(X\mid(Y,Z))$.

The \textit{chain rule for entropy} states that
\[
 H_q(X,Y)=H_q(Y)+H_q(X\mid Y).
\]
Its conditional form is
\[
 H_q(X,Y\mid Z)
 =H_q(X\mid Z)+H_q(Y\mid X,Z).
\]

Let $X$ and $Y$ be random variables with joint distribution
$P_{XY}$ and marginals $P_X$ and $P_Y$. Their mutual information is
defined by
\[
I_q(X;Y) := \sum_{x,y} P_{XY}(x,y) \log_q \frac{P_{XY}(x,y)}{P_X(x)P_Y(y)}.
\]
Equivalently,
\[
I_q(X;Y) = H_q(X)-H_q(X \mid Y) = H_q(Y)-H_q(Y \mid X) = H_q(X)+H_q(Y)-H_q(X,Y).
\]
Terms with $P_{XY}(x,y)=0$ contribute zero. Moreover, $I_q(X;Y) \geq 0$, and $I_q(X;Y)=0$ if and only if $X$ and $Y$ are independent.

\paragraph{Conditional mutual information and the chain rule.}
For random variables $X,Y,Z$, define
\[
 I_q(X;Y\mid Z)
 :=
 H_q(X\mid Z)-H_q(X\mid Y,Z)
 =
 H_q(Y\mid Z)-H_q(Y\mid X,Z).
\]
It is nonnegative, with equality exactly when $X$ and $Y$ are
conditionally independent given $Z$, meaning that
\[
 P_{XY\mid Z=z}(x,y)
 =
 P_{X\mid Z=z}(x)P_{Y\mid Z=z}(y)
\]
for every $z$ of positive probability.

The \textit{chain rule for mutual information} is
\[
 I_q(X;Y,Z)
 =I_q(X;Y)+I_q(X;Z\mid Y)
 =I_q(X;Z)+I_q(X;Y\mid Z).
\]
Here $I_q(X;Y,Z)$ denotes mutual information between $X$ and
the pair $(Y,Z)$. Nonnegativity of conditional mutual information
also gives
\[
 H_q(X\mid Y,Z)\leq H_q(X\mid Y)\leq H_q(X).
\]

\paragraph{Data processing inequality.}
We write $X\to Y\to Z$ when $X$ and $Z$ are conditionally
independent given $Y$. Equivalently, their joint distribution
factorizes as
\[
 P_{XYZ}(x,y,z)
 =P_X(x)P_{Y\mid X=x}(y)P_{Z\mid Y=y}(z);
\]
which is called the \emph{Markov-chain condition}. 

The data-processing inequality states that
\begin{equation} \label{eq:dataProcessing}
 I_q(X;Z)\leq I_q(X;Y) \qquad \text{for} \qquad X\to Y\to Z.
\end{equation}
In particular, $I_q(X;f(Y))\leq I_q(X;Y)$ for every deterministic
function $f$, with equality when $f$ is bijective.

\paragraph{Convexity in the channel.}
For a fixed input distribution $p$, mutual information is convex
in the conditional distribution of the output
\cite[Theorem~2.7.4]{CoverThomas2006}.
Explicitly, let $X\sim p$, and let
$P_{Y_i\mid X}=W_i$, $i=1,2$, be channels with the same
input and output alphabets. If
\[
 P_{Y\mid X} = \lambda W_1+(1-\lambda)W_2, \qquad 0 \leq \lambda \leq 1,
\]
then
\[
 I_q(X;Y)
 \leq
 \lambda I_q(X;Y_1)+(1-\lambda)I_q(X;Y_2).
\]
The same statement holds for any finite convex combination of channels.

\paragraph{Majorization and Schur concavity.}
We say that a nonincreasing probability vector $y$ \textit{majorizes} nonincreasing probability vector $x$, and write $x \preceq y$ when
\[
 \sum_{i=1}^j x_i \leq \sum_{i=1}^j y_i 
 \quad(1\leq j < k),
\]
where $k$ is the length of $x$ and $y$.
A permutation-invariant function $f$ on probability vectors of length $k$ is \textit{Schur-concave} if $f(x) \geq f(y)$ provided that $x\preceq y$.
Entropy function $H_q$ is strictly Schur-concave, which means that $H_q(x) = H_q(y)$ and $x\preceq y$ imply $x=y$.
Also, it is concave in the usual sense, that is
\[
H_q(\lambda p_1 + (1-\lambda) p_2) \geq \lambda H_q(p_1) + (1-\lambda) H_q(p_2)
\]
for all probability vectors $p_1,p_2$ and $0 \leq \lambda \leq 1$.

\subsection{Asymptotic case of Theorem~\ref{thm:finite-permutation-transport}}

The passage from Theorem~\ref{thm:finite-permutation-transport} to an
asymptotic rate is a standard method-of-types calculation~\cite{Csiszar1998}: multinomial
coefficients become entropies, and a quotient of conditional and marginal
type-class sizes becomes a mutual information; see, for example,
\cite{Blahut1977,CoverThomas2006}. This observation is particularly useful
here because it turns the combinatorial choice of an auxiliary permutation
sequence into a finite-dimensional information-theoretic optimization.

Let $\Pi$ be a random variable on the symmetric group $\Perm_q$, let $(X,\Pi)$ have joint distribution $\eta$, and let $Y=\Pi(X)$. Stirling's
formula applied to \eqref{eq:finite-permutation-transport} gives the following
form of Theorem~\ref{thm:finite-permutation-transport}.
As usual, we write $X \sim Z$ for equidistributed random variables $X$ and $Z$.

\begin{corollary}
\label{cor:asymptotic-permutation-transport}
For every joint distribution of $(X,\Pi)$ with $X\sim p$ and $Y=\Pi(X)\sim r$,
\begin{equation}
 R_q(\delta;p)
 \leq I_q(X;\Pi)+R_q(\delta;r),
 \label{eq:mutual-information-transport}
\end{equation}
where $I_q$ denotes mutual information measured in base-$q$ units.
\end{corollary}

Indeed, the exponent of the multinomial ratio is
$H_q(\Pi)-H_q(\Pi \mid X)=I_q(X;\Pi)$.

\paragraph{Relation with conditional-type Elias arguments.} The appearance of mutual information in \eqref{eq:mutual-information-transport} is not specific to the permutation formulation. It is the standard exponential loss in conditional-type versions of the Bassalygo--Elias argument. Indeed, Blahut's composition bounds~\cite{Blahut1977} are formulated using stochastic matrices subject to mutual-information constraints. In a particularly transparent modern formulation, Dalai's conditional-composition covering lemma~\cite[Lemma~3]{Dalai2015} shows that a constant-composition code $C$ contains a subcode $C'$ of size 
\[ |C'| \geq |C|\,q^{-n(I_q(P,V)+o(1))} \]
whose codewords have a prescribed joint type with a fixed auxiliary composition. 
Here $I_q(P,V)=I_q(X;U)$, where $X\sim P$ and $U \mid X=x\sim V(\cdot \mid x)$.
The finite permutation inequality used here gives an exact group-action realization of the same conditional-type mechanism.

We shall continue investigation of the relations of our approach with the Bassalygo--Elias argument in Subsection~\ref{sec:master-BE}.

\subsection{Optimal transport between arbitrary compositions}
\label{sec:cc-exact-transport}

The permutation inequality~\eqref{eq:mutual-information-transport} also permits a target whose nonzero symbol
frequencies are not equal. Introduce the minimal transport cost as
\begin{equation}
 K_q(p,r):=\min_{P_{\Pi \mid X}:\,\Pi(X)\sim r} I_q(X;\Pi),
 \qquad X\sim p,\quad\Pi\in\Perm_q.
 \label{eq:cc-cost}
\end{equation}
Thus the strongest form of~\eqref{eq:mutual-information-transport} becomes $R_q(\delta;p)\leq K_q(p,r)+R_q(\delta;r)$.
The minimum in~\eqref{eq:cc-cost} is attained, since the feasible set of joint distributions is
compact and mutual information is continuous. The next result (Theorem~\ref{thm:cc-join-cost}) solves~\eqref{eq:cc-cost} and thus
eliminates optimization over the permutation-valued auxiliary variable $\Pi$.

For a vector $v$ define $v^\downarrow$ as the permuted vector whose entries go in the non-increasing order.
We say that a probability vector $s$ \textit{majorizes} a probability vector $p$ if $p^\downarrow \preceq s^\downarrow$. 
From combinatorial point of view the coordinates order in the composition vector does not matter, so we will simply write $p \preceq s$.
Majorization gives a standard partial order on probability vectors. Denote the least common majorant of vectors $p$ and $r$ in the corresponding lattice by
$p^\downarrow\vee r^\downarrow$. The majorization lattice and its entropy functionals have been studied
in~\cite{CicaleseVaccaro2002,cicalese2013information}.
In particular, the symmetric quantity
\begin{equation} \label{eq:metrics}
 \mathsf d_q(p,r) := H_q(p)+H_q(r)-2H_q(p^\downarrow\vee r^\downarrow) 
\end{equation}
defines a metric on nonincreasing probability vectors.

For a permutation $\sigma\in\Perm_q$ and a probability vector $v$,
write $(\sigma (v))_x:=v_{\sigma^{-1}(x)}$.
Thus, if $S\sim v$, then $\sigma(S)$ has distribution $\sigma(v)$.
We use the standard characterization of majorization:
\begin{equation} \label{eq:convexRepresentation}
 u\preceq v
 \quad\Longleftrightarrow\quad
 u\in\operatorname{conv}\{\sigma (v):\sigma\in\Perm_q\}.
\end{equation}
This follows from the doubly stochastic characterization of majorization
and the Birkhoff--von Neumann theorem; see
\cite[Chapters~2 and~4]{MarshallOlkinArnold2011}.

\begin{theorem}[Exact constant-composition transport cost]
\label{thm:cc-join-cost}
Let $p,r$ be probability vectors on the $q$-letter alphabet $\A$,
and put $s:=p^\downarrow\vee r^\downarrow$.
Then
\begin{equation} \label{eq:cc-join-cost}
 K_q(p,r)=H_q(p)-H_q(s).
\end{equation}
Moreover, let 
\[
 p=\sum_i\alpha_i\sigma_i (s),
 \qquad
 r=\sum_j\beta_j \rho_j (s),
\]
where $\sigma_i,\rho_j \in\Perm_q$,
$\alpha_i,\beta_j\geq0$, and
$\sum_i\alpha_i=\sum_j\beta_j=1$, be arbitrary convex representations of $p$ and $r$ (such representations exist because $p,r\preceq s$).
Then an optimal joint distribution of $(X,\Pi)$ is
\begin{equation} \label{eq:optmialJoint}
 \Pr\{X=x,\Pi=\pi\} = \sum_{i,j :\, \rho_j\circ\sigma_i^{-1}=\pi} \alpha_i\beta_j s_{\sigma_i^{-1}(x)}.
\end{equation}
\end{theorem}

\begin{proof}
Consider any feasible joint distribution of $(X,\Pi)$, with
$X\sim p$ and $Y=\Pi(X)\sim r$. Put $\mu_\pi:=\Pr\{\Pi=\pi\}$. For each $\pi$ with
$\mu_\pi>0$, define
\[
 v_{a,\pi}:=\Pr\{X=a\mid\Pi=\pi\},
 \qquad
 v_\pi:=(v_{a,\pi})_{a\in\A}.
\]
Then, for every $a,b\in\A$,
\[
 p_a =\Pr\{X=a\} =\sum_{\pi\in\Perm_q}\mu_\pi v_{a,\pi}, \qquad 
 r_b =\Pr\{Y=b\} =\sum_{\substack{a\in\A,\;\pi\in\Perm_q\\\pi(a)=b}} \mu_\pi v_{a,\pi},
\]
and
\[
 H_q(X\mid\Pi)
 =-\sum_{\pi\in\Perm_q}\mu_\pi
 \sum_{a\in\A}v_{a,\pi} \log_q v_{a,\pi}
 =\sum_{\pi\in\Perm_q}\mu_\pi H_q(v_\pi).
\]

Set
\[
 s_0:=\sum_\pi\mu_\pi v_\pi^\downarrow.
\]
This is a non-increasing probability vector. For each $k$, the function $F_k(v):=\sum_{i=1}^k v_i^\downarrow$
is convex, since it is the maximum of the sums of $k$ coordinates.
Therefore
\[
 F_k(p)
 \leq\sum_\pi\mu_\pi F_k(v_\pi)
 =F_k(s_0).
\]
The same argument applied to the mixture defining $r$, together with
$F_k(\pi (v_\pi))=F_k(v_\pi)$, gives $F_k(r)\leq F_k(s_0)$.
Thus $p,r\preceq s_0$, and the defining property of the join implies
$s\preceq s_0$.

By concavity and Schur concavity of entropy,
\[
 H_q(X\mid\Pi) = \sum_\pi\mu_\pi H_q((v_\pi)^\downarrow) \leq H_q(s_0) \leq H_q(s).
\]
Hence every feasible transport satisfies $I_q(X;\Pi)\geq H_q(p)-H_q(s)$.

For the reverse inequality, take convex representations from~\eqref{eq:convexRepresentation}. 
Let $S,I,J$ be mutually independent, with
\[
 \Pr\{S=x\}=s_x,
 \qquad
 \Pr\{I=i\}=\alpha_i,
 \qquad
 \Pr\{J=j\}=\beta_j,
\]
and define
\[
 X:=\sigma_I(S),
 \qquad
 \Pi:=\rho_J\circ\sigma_I^{-1},
 \qquad
 Y:=\Pi(X)=\rho_J(S).
\]
The two convex representations give $X\sim p$ and $Y\sim r$.
This construction induces exactly the joint distribution stated above.

Since $\Pi$ is a deterministic function of $(I,J)$, data processing inequality~\eqref{eq:dataProcessing}
gives $I_q(X;\Pi)\leq I_q(X;I,J)$.
Conditionally on $I=i,J=j$, the distribution of $X$ is
$\sigma_i (s)$, whose entropy is $H_q(s)$. Consequently,
\[
 I_q(X;\Pi)
 \leq
 H_q(p)-H_q(X \mid I,J)
 =
 H_q(p)-H_q(s).
\]
Together with the lower bound, this proves the formula and the optimality
of the constructed plan.
\end{proof}

\paragraph{The least-concave-majorant algorithm.}
Here we present a standard algorithm to find $s = p^\downarrow\vee r^\downarrow$ for given $p$ and $r$. The algorithm is linear in $q$ after the sorting 
(which has complexity $O(q\log q)$).

Set $m_0=0$, $m_q=1$, and
\[
 m_k=\max\left\{\sum_{i=1}^k p_i^\downarrow,
\sum_{i=1}^k r_i^\downarrow\right\}.
\]
Start with the increments $m_k-m_{k-1}$ as unit-length blocks.
Whenever two adjacent block slopes increase, merge the blocks and
replace their slopes by their length-weighted average.
Continue until the block slopes are nonincreasing.
This gives the least concave majorant $\widehat m$ of the polygonal line
through $(k,m_k)$. The join has coordinates
$s_i=\widehat m(i)-\widehat m(i-1)$. Once the join and the two convex decompositions have been computed,
the realizing plan is given by \eqref{eq:optmialJoint}.

For $q=3$ the least common majorant can be written explicitly. Let $M$ and $m$ be the maximum and the minimum among all entries of $p$ and $r$, respectively.
Then
\[
 p^\downarrow\vee r^\downarrow=(M,1-M-m,m),\qquad
 K_3(p,r)=H_3(p)-H_3(M,1-M-m,m).
\]
Indeed, the largest coordinate of any common majorant is at least $M$,
and its smallest coordinate is at most $m$. The vector $p^\downarrow\vee r^\downarrow$ is
non-increasing and satisfies both requirements.

\begin{corollary}
\label{cor:cc-zero-cost}
The cost~\eqref{eq:cc-join-cost} is zero exactly when $r\preceq p$. Also,
$K_q(p,r)=H_q(p)-H_q(r)$ if and only if $p\preceq r$.
\end{corollary}

\subsection{Optimal transport between constant-weight layers}
\label{sec:weight-transport}

In the next section (Proposition~\ref{prop:uniform-composition}) we shall prove that from the asymptotic point of view the constant-weight case coincides with the symmetric composition case $p_t=\left(1-t,\frac{t}{q-1},\ldots,\frac{t}{q-1}\right)$ for $0 \leq t \leq 1$. Now we study the optimal transport between symmetric source and target.

Recall that the preceding least-concave-majorant algorithm computes the transport cost. 
For symmetric compositions, the computation is particularly simple and can be done explicitly. 
For 
$1\leq k\leq q-1$, the sum of the $k$ largest components of $p_t$ is
\begin{equation}
 P_k(t):=\sum_{i=1}^k(p_t^\downarrow)_i
 =\begin{cases}
 1-t+\dfrac{(k-1)t}{q-1},&0\leq t\leq(q-1)/q,\\[2mm]
 \dfrac{kt}{q-1},&(q-1)/q\leq t\leq1.
 \end{cases}
 \label{eq:cw-partial-sums}
\end{equation}
The algorithm is applied to $\max\{P_k(\omega),P_k(\theta)\}$,
with the endpoint values $0$ and $1$. Once the join is known, an
optimal transport plan is obtained by the construction in
Theorem~\ref{thm:cc-join-cost}.

We call the intervals on which the profile has the explicit expressions
below \textit{its branches}; the \textit{first branch} is the one adjacent to
$\theta=0$. The branches can now be determined by comparing the partial
sums in \eqref{eq:cw-partial-sums} and pooling their slopes as in the
least-concave-majorant algorithm.

The terminology for the first branch comes from Fano's inequality
\cite[Section~2.10]{CoverThomas2006}, see~\eqref{eq:fano-lower-bound-J} below. For a feasible transport with $X\sim p_\omega$ and
$\Pi(X)\sim p_\theta$, put $Z:=\Pi^{-1}(0)$.
Then $Z\in\A$, $\Pr\{X\neq Z\}=\theta$, and data processing inequality~\eqref{eq:dataProcessing} gives
$I_q(X;Z)\leq I_q(X;\Pi)$.
If
$E=\mathbf1_{\{X\neq Z\}}$ and $\Pr\{E=1\}=\theta$, then
\[
 H_q(X \mid Z)
 =H_q(E \mid Z)+H_q(X \mid E,Z)
 \leq H_q(E)+\theta\log_q(q-1)=h_q(\theta).
\]
Indeed, when $E=0$, the value of $X$ is determined by $Z$,
whereas when $E=1$, at most $q-1$ possibilities remain. Therefore
\begin{equation}
 I_q(X;Z)\geq h_q(\omega)-h_q(\theta), \qquad X\sim p_\omega.
 \label{eq:fano-lower-bound-J}
\end{equation}

The next proposition gives a compact expression for the entire profile. The following Corollary~\ref{cor:explicit-profile} identifies its branches explicitly. In particular, it shows that Fano's lower bound is sharp on the initial interval, see~\eqref{eq:J-high-weight} below.

\begin{proposition}[Complete constant-weight profile]
\label{prop:complete-high-weight}
Let $q\geq3$ and $\omega,\theta\in[0,1]$. Put
\begin{equation} \label{def:cases}
M :=\max\left\{1-\omega,1-\theta,
 \frac{\omega}{q-1},\frac{\theta}{q-1}\right\}, \quad \ell:=\min\left\{1-\omega,1-\theta,
 \frac{\omega}{q-1},\frac{\theta}{q-1}\right\},
\end{equation}
and $\lambda:=(1-M-\ell)/(q-2)$. Then
\begin{equation}
 p_\omega^\downarrow\vee p_\theta^\downarrow
 =
 \left(
 M,\underbrace{\lambda,\ldots,\lambda}_{q-2\text{ entries}},\ell
 \right),
 \label{eq:high-weight-join}
\end{equation}
and consequently
\begin{equation} \label{eq:complete-cw-cost}
 K_q(p_\omega,p_\theta)
 =
 h_q(\omega)+M\log_q M
 +(q-2)\lambda\log_q\lambda+\ell\log_q\ell.
\end{equation}

\end{proposition}

\begin{proof}
For $t\in\{\omega,\theta\}$, let $P_k(t)$ be as defined in~\eqref{eq:cw-partial-sums}.
Clearly, these partial sums are linear in $k$ on $1\leq k\leq q-1$, regardless of which side of $(q-1)/q$
contains $t$. Moreover,
\[
 M=\max\{P_1(\omega),P_1(\theta)\},
 \qquad
 1-\ell=\max\{P_{q-1}(\omega),P_{q-1}(\theta)\}.
\]
We compute the least concave majorant directly.

A composition attaining the smallest component $\ell$ has all its
other components at most $M$, while a composition attaining the
largest component $M$ has all its other components at least $\ell$.
Thus $1-\ell\leq(q-1)M$ and $1-M\geq(q-1)\ell$.
These inequalities give 
\[
M \geq \lambda \geq \ell \geq 0,
\]
so $s=(M,\lambda,\ldots,\lambda,\ell)$
is a non-increasing probability vector. Its partial sums are
\[
 S_k = M+(k-1)\lambda = \frac{q-1-k}{q-2}M +
 \frac{k-1}{q-2}(1-\ell), \qquad 1\leq k\leq q-1.
\]
Since each $P_k(t)$ is linear in $k$ and its two endpoint values
are bounded above by $M$ and $1-\ell$, respectively,
\[
 P_k(\omega)\leq S_k,
 \qquad
 P_k(\theta)\leq S_k.
\]
Hence $s$ majorizes both compositions.

Conversely, let $u = (u_0,\dots,u_{q-1})$ be any non-increasing common majorant of $p_\omega$ and $p_\theta$.
Denote by $U_k$, $1 \leq k \leq q$, the partial sums of $u$.
Then $U_1\geq M$ and $U_{q-1}\geq1-\ell$. This and the concavity of its
partial sums imply
\[
 U_k
 \geq
 \frac{q-1-k}{q-2}U_1
 +
 \frac{k-1}{q-2}U_{q-1}
 \geq S_k,
 \qquad 1\leq k\leq q-1.
\]
Therefore $s\preceq u$, proving \eqref{eq:high-weight-join}.
Theorem~\ref{thm:cc-join-cost} now gives
\eqref{eq:complete-cw-cost} thus completing the proof.
\end{proof}

Inspecting the cases in~\eqref{def:cases} gives the following explicit representation of Proposition~\ref{prop:complete-high-weight}.

\begin{corollary} \label{cor:explicit-profile}
For $\omega>0$ and $q \geq 3$, put
$t_0=(q-1)(1-\omega)$, $t_1=1-\frac{\omega}{q-1}$, $t_2=\omega$,
and define
\[
 \bar h_{q-1}^{(q)}(z)
 :=
 -z\log_q z-(1-z)\log_q(1-z)+z\log_q(q-2),
 \qquad 0\leq z\leq1.
\]
Then the complete profile is
\begin{equation}
 K_q(p_\omega,p_\theta)
 =
 \begin{cases}
 h_q(\omega)-h_q(\theta),
 &\begin{array}{@{}l@{}}
 \theta\leq\min\{t_0,t_2\}\\[-1mm]
 \text{or }\theta\geq\max\{t_0,t_2\},
 \end{array}\\[3mm]
 0,
 &\min\{t_1,t_2\}\leq\theta\leq\max\{t_1,t_2\},\\[2mm]
 \displaystyle
 \omega\left[
 \log_q(q-1)-
 \bar h_{q-1}^{(q)}
 \!\left(\frac{\theta-1+\omega}{\omega}\right)
 \right],
 &\text{otherwise}.
 \end{cases}
 \label{eq:J-high-weight}
\end{equation}
All cases are restricted to $\theta\in[0,1]$, and the expressions
agree at common endpoints. If $t_0>1$, the condition
$\theta\geq\max\{t_0,t_2\}$ gives an empty branch.
For $\omega=0$, and any $\theta$ one has $K_q(p_0,p_\theta)=0$.
\end{corollary}

\begin{proof}
To obtain the piecewise expression, the same partial-sum comparison
gives
\begin{align*}
 p_\omega\preceq p_\theta
 &\quad\Longleftrightarrow\quad
 \theta\leq\min\{t_0,t_2\}
 \ \text{or}\ 
 \theta\geq\max\{t_0,t_2\},\\
 p_\theta\preceq p_\omega
 &\quad\Longleftrightarrow\quad
 \min\{t_1,t_2\}\leq\theta\leq\max\{t_1,t_2\}.
\end{align*}
In these two regimes the join is, respectively,
$p_\theta^\downarrow$ and $p_\omega^\downarrow$, giving the first
and second expressions in \eqref{eq:J-high-weight}.

In every remaining case the compositions are incomparable, and Proposition~\ref{prop:complete-high-weight} implies that their
join is a permutation of $(1-\omega, 1-\theta, u, \dots, u)$, with $u := (\omega+\theta-1)/(q-2)$ appearing $q-2$ times.
Put $z:= (\theta-1+\omega)/\omega$.
Here $0\leq z\leq1$, and the same vector can be written, up to
permutation, as $(1-\omega, \omega(1-z), \omega z/(q-2),\ldots, \omega z/(q-2))$.

Recall that $s = p_\omega^\downarrow\vee p_\theta^\downarrow$; splitting off the mass $1-\omega$ in the entropy yields
\[
 H_q(p_\omega)-H_q(s)
 =
 \omega\left[
 \log_q(q-1)-\bar h_{q-1}^{(q)}(z)
 \right],
\]
which proves the remaining expression. At the boundaries the respective
join vectors coincide, so the formulas agree. Finally, when
$\omega=0$, the join is $(1,0,\ldots,0)$; hence
\eqref{eq:complete-cw-cost} gives zero.
\end{proof}

The ordering of the branches in~\eqref{eq:J-high-weight} depends on the source weight $\omega$. 
For $\omega>(q-1)/q$, one has $t_0<t_1<t_2$, and the incomparable
regime ``otherwise'' lies before the zero-information interval. For
$0<\omega<(q-1)/q$, the order is $t_2<t_1<t_0$, and the
incomparable regime lies after that interval. The former is the classical
reverse-water-filling regime; the latter belongs to the positive
exact-distortion continuation beyond the zero-information interval (see Appendix A).
At $\omega=(q-1)/q$, for any $\theta$ the three thresholds coincide and
\[
 K_q(p_{(q-1)/q},p_\theta)=1-h_q(\theta).
\]

\begin{remark}
In the case $q=2$, the two compositions are always comparable in the majorization order, and
\[
 K_2(p_\omega,p_\theta) =\max\{h_2(\omega)-h_2(\theta),0\}.
\]
\end{remark}

In the following proposition we verify that if both source and target compositions are symmetric then Theorem~\ref{thm:finite-permutation-transport} cannot asymptotically refine Theorem~\ref{thm:finite-matrix-transport}.

\begin{proposition}
\label{prop:additiveShifts}
For symmetric compositions $p_\omega,p_\theta$, an optimal
permutation transport can also be realized by additive shifts of
$\mathbb Z_q$.
\end{proposition}

\begin{proof}
For $q=2$, every alphabet permutation is an additive shift, so
assume $q\geq3$. Choose an optimal permutation transport
$(X,\Pi)$ from $p_\omega$ to $p_\theta$, as provided by
Theorem~\ref{thm:cc-join-cost}, and put $Z:=\Pi^{-1}(0)$.
Since $X=Z$ if and only if $\Pi(X)=0$, data processing inequality~\eqref{eq:dataProcessing} gives
\[
 \Pr\{X=Z\}=1-\theta,
 \qquad
 I_q(X;Z)\leq I_q(X;\Pi)=K_q(p_\omega,p_\theta).
\]

Average the joint distribution of $(X,Z)$ under simultaneous
alphabet permutations fixing $0$, and denote the resulting pair
by $(X',Z')$. These relabelings preserve the source distribution
$p_\omega$ and the event of equality. By invariance under
bijective relabeling and convexity of mutual information in the
channel for a fixed source distribution,
\[
 X'\sim p_\omega,\qquad
 \Pr\{X'=Z'\}=1-\theta,\qquad
 I_q(X';Z')\leq K_q(p_\omega,p_\theta).
\]

Write $Q_{xz}:=\Pr\{X'=x,Z'=z\}$.
By symmetry, there are constants $A,B,C$ such that
\[
 Q_{0j}=A,\qquad Q_{i0}=B,\qquad Q_{ij}=C
 \quad(i,j\neq0,\ i\neq j).
\]
For every $u\neq0$, subtraction in $\mathbb Z_q$ gives
\[
 \Pr\{X'-Z'=u\}
 =\sum_{z\in\mathbb Z_q}Q_{z+u,z}
 =A+B+(q-2)C.
\]
Thus all nonzero differences have the same probability. Since
$\Pr\{X'-Z'=0\}=1-\theta$, we obtain
$X'-Z'\sim p_\theta$.
Then the random translation $x\mapsto x-Z'$ is a feasible permutation transport.

Consequently,
\[
 K_q(p_\omega,p_\theta) \leq I_q(X';-Z') = I_q(X';Z') \leq K_q(p_\omega,p_\theta).
\]
Hence the additive transport is optimal.
\end{proof}

\section{The closure operator and its properties}
\label{sec:closure}

In this section we define the closure operator, investigate its idempotence and prove necessary, and sufficient conditions on fixed points of the operator.

Let $\Delta_{q-1}$ denote the probability simplex on $\A$, that is the convex hull of standard basis vectors in $\mathbb{R}^q$.
Let
$B_q(\delta;r)$ be a valid asymptotic upper bound for every composition $r$. 
Define the constant-composition closure operator 
\[
 (\Pp B_q)(\delta;p) :=\inf_{r\in\Delta_{q-1}} \{K_q(p,r)+B_q(\delta;r)\}.
\]

Theorem~\ref{thm:cc-join-cost} has already optimized the information cost
of transferring a bound between two prescribed compositions. Thus the
remaining optimization is over the target composition $r$ alone:
\[
(\Pp B_q)(\delta;p) = H_q(p)+\inf_{r\in\Delta_{q-1}} \{B_q(\delta;r)-H_q(p^\downarrow\vee r^\downarrow)\}.
\]

\subsection{Idempotence and majorization}

The following theorem follows from a certain triangle-type inequality and implies that no finite chain of intermediate compositions improves a single application of the presented method.

\begin{theorem}[Idempotence] \label{thm:idempotence}
The constant-composition closure operator is idempotent, that is for every upper bound $B_q$ we have
\[
 \Pp(\Pp B_q)=\Pp B_q.
\]
\end{theorem}

\begin{proof}
We first prove that, for any compositions $p,r,s$,
\begin{equation}
 K_q(p,s)\leq K_q(p,r)+K_q(r,s).
 \label{eq:cc-cost-triangle}
\end{equation}
Choose optimal transports $X\stackrel{\Pi}{\longmapsto}Y$ and
$Y\stackrel{\Sigma}{\longmapsto}Z$, with respective marginals
$p,r,s$. Realize the second transport so that $\Sigma$ is
conditionally independent of $(X,\Pi)$ given $Y$, and put
$\Theta=\Sigma\circ\Pi$. The chain rule and data processing inequality~\eqref{eq:dataProcessing} give
\[
 I_q(X;\Theta) \leq I_q(X;\Pi,\Sigma) =I_q(X;\Pi)+I_q(X;\Sigma \mid \Pi) =
\]
\[
= I_q(X;\Pi)+I_q(Y;\Sigma \mid \Pi) \leq I_q(X;\Pi)+I_q(Y;\Sigma).
\]
Here we used the conditional bijection $Y=\Pi(X)$, the fact that $H_q(\Sigma \mid \Pi)\leq H_q(\Sigma)$ and
$H_q(\Sigma \mid Y,\Pi)=H_q(\Sigma \mid Y)$. Since $\Theta$
transports $p$ to $s$, this proves \eqref{eq:cc-cost-triangle}.

Consequently,
\begin{align*}
 (\Pp(\Pp B_q))(\delta;p)
 &=\inf_{r,s\in\Delta_{q-1}}
 \{K_q(p,r)+K_q(r,s)+B_q(\delta;s)\}\\
 &\geq\inf_{s\in\Delta_{q-1}}
 \{K_q(p,s)+B_q(\delta;s)\}
 =(\Pp B_q)(\delta;p).
\end{align*}
Choosing $r=p$ gives the reverse inequality. 
\end{proof}

\begin{proposition}[Uniform nonzero composition]
\label{prop:uniform-composition}
For every $\omega$,
\[
R_q(\delta;p_\omega)=R_q(\delta,\omega).
\]
\end{proposition}

\begin{proof}
For any composition vector $p=(1-\omega,p_1,\ldots,p_{q-1})$
choose $\Pi$, independently of $X$, uniformly from the subgroup
$\Perm_{q-1}\leq \Perm_q$ of permutations with fixed $0$. Then
\[
I_q(X;\Pi)=0 \qquad \text{ and } \qquad \Pi(X)\sim p_\omega. 
\]
Corollary~\ref{cor:asymptotic-permutation-transport} gives
$R_q(\delta;p)\leq R_q(\delta;p_\omega)$
uniformly over all source compositions of weight $\omega$. 

By inclusion, we have
\[
A_q(n,d,w) \leq \sum_{\wt(\bm w) = w} A_q(n,d,\bm w).
\]
Since the number of summands is polynomial in $n$, passing to the limit gives 
\[
R_q(\delta, \omega) \leq \max_{p, p_1+\dots+ p_{q-1}=\omega} R_q(\delta;p) = R_q(\delta;p_\omega).
\]
The inverse inequality is obvious by inclusion.
\end{proof}

Now we are in a position to show that $R_q(\delta,\omega)$ is unimodal in $\omega$, i.e. for every fixed $\delta$ function $R_q(\delta,\omega)$ increases before $\omega = (q-1)/q$ and decreases later. In fact, Corollary~\ref{cor:fp-exact-rates} gives a stronger result of the behavior of $R_q(\delta,\omega)$.
In the binary case, the unimodality results were obtained by Samorodnitsky~\cite{Samorodnitsky2001}, see also~\cite{BachocEtAl2011}.

\begin{corollary}[Majorization monotonicity]
\label{cor:majorization}
If $r\preceq p$, then
\begin{equation}
 R_q(\delta;p)\leq R_q(\delta;r).
 \label{eq:majorization-monotonicity}
\end{equation}
Thus the asymptotic constant-composition rate is Schur-concave. In
particular, for fixed $\delta$, the function
$\omega\mapsto R_q(\delta,\omega)$ is nondecreasing on
$[0,(q-1)/q]$ and nonincreasing on $[(q-1)/q,1]$.
\end{corollary}

\begin{proof}
Since $r\preceq p$, Corollary~\ref{cor:cc-zero-cost} gives $K_q(p,r)=0$. 
The asymptotic transport inequality~\eqref{eq:mutual-information-transport} therefore yields~\eqref{eq:majorization-monotonicity}.

Writing $\omega_*=(q-1)/q$, comparison of partial sums~\ref{eq:cw-partial-sums} gives $p_{\omega_2}\preceq p_{\omega_1}$ for $0\leq\omega_1\leq\omega_2\leq\omega_*$ and 
$p_{\omega_1}\preceq p_{\omega_2}$ for $\omega_*\leq\omega_1\leq\omega_2\leq1$.
The stated weight monotonicity follows, using Proposition~\ref{prop:uniform-composition} for the passage from symmetric
compositions to full constant-weight shells.
\end{proof}

\subsection{Bounds unchanged by permutation transport}
\label{sec:transport-fixed-points}

Throughout this subsection, $q\geq2$ and $\delta\in[0,1]$ are fixed.
We characterize finite-valued upper bounds that cannot be improved at
any source by the permutation-transport operator. The criteria are
pointwise in $\delta$ and impose no relation between different distances.

We will use the metric $\mathsf d_q(p,r)$, defined in~\eqref{eq:metrics}, on compositions modulo coordinate permutations. 

\begin{theorem}[Unimprovable composition-dependent bounds]
\label{thm:fp-composition}
Let $B_q(\delta;\cdot)$ be an upper bound on
$R_q(\delta;\cdot)$. The following conditions are equivalent.
\begin{enumerate}
\item For every $p\in\Delta_{q-1}$,
\[
 (\Pp B_q)(\delta;p)=B_q(\delta;p).
\]
\item Both functions
\[
 p\longmapsto B_q(\delta;p),
 \qquad
 p\longmapsto H_q(p)-B_q(\delta;p)
\]
are Schur-concave.
\item For every $p\preceq r$,
\begin{equation}
 0\leq B_q(\delta;p)-B_q(\delta;r)
 \leq H_q(p)-H_q(r).
 \label{eq:fp-comparable}
\end{equation}
\item For every $p,r\in\Delta_{q-1}$,
\begin{equation}
 \left|
 \bigl(2B_q(\delta;p)-H_q(p)\bigr)
 -\bigl(2B_q(\delta;r)-H_q(r)\bigr)
 \right|
 \leq\mathsf d_q(p,r).
 \label{eq:fp-composition-lipschitz}
\end{equation}
\end{enumerate}
\end{theorem}

\begin{proof}
Since $K_q(p,p)=0$, condition~1
is equivalent to the family of inequalities
\begin{equation}
 B_q(\delta;p)-B_q(\delta;r) \leq K_q(p,r)
 \qquad(p,r\in\Delta_{q-1}).
 \label{eq:fp-directed}
\end{equation}
Indeed, these inequalities give $B\leq\Pp B$, whereas the diagonal
choice $r=p$ always gives $\Pp B\leq B$.

If $p\preceq r$, the join formula yields
\[
 K_q(p,r) = H_q(p)-H_q(r),\qquad K_q(r,p)=0.
\]
Applying~\eqref{eq:fp-directed} in both directions proves
\eqref{eq:fp-comparable}. Conditions~2 and~3 are equivalent by the
definition of Schur-concavity. Notice that mutually majorizing vectors
give equality in~\eqref{eq:fp-comparable}, so permutation invariance is
also a consequence of condition~3.

Conversely, assume condition~3. For arbitrary $p,r$, put
$s=p^\downarrow\vee r^\downarrow$. Since $p,r\preceq s$,
\[
 B_q(\delta;p) - B_q(\delta;r)
 \leq B_q(\delta;p) - B_q(\delta;s)
 \leq H_q(p)-H_q(s)
 =K_q(p,r).
\]
Thus~\eqref{eq:fp-directed} holds and condition~1 follows.

Finally,
\[
 2K_q(p,r)=\mathsf d_q(p,r)+H_q(p)-H_q(r).
\]
Consequently, the two directed inequalities for $(p,r)$ and $(r,p)$ in~\eqref{eq:fp-directed}
are equivalent to~\eqref{eq:fp-composition-lipschitz}, that is condition 4.
\end{proof}

For a constant-weight input bound $B_q(\delta,\theta)$ we use the extension
$B_q(\delta;r)=B_q(\delta,1-r_0)$, where $r_0$ is the first coordinate of $r$.
At a symmetric source the closure is
\begin{equation}
 (\Pp B_q)(\delta;p_\omega)
 =\inf_{0\leq\theta\leq1}
 \{K_q(p_\omega,p_\theta)+B_q(\delta,\theta)\}.
 \label{eq:fp-weight-closure}
\end{equation}
Thus unimprovability of a weight-only bound means equality in
\eqref{eq:fp-weight-closure} with $B_q(\delta,\omega)$ for every
$\omega$. It concerns the symmetric sources $p_\omega$; it does not
assert that the extension $r\mapsto B_q(\delta,1-r_0)$ is a fixed point
on the entire simplex.

\begin{theorem}[Unimprovable weight-dependent bounds]
\label{thm:fp-weight}
Let $B_q(\delta,\cdot)$ be an upper bound on
$R_q(\delta,\cdot)$, and put
\[
 f(\omega)=B_q(\delta,\omega),\qquad
 h(\omega)=h_q(\omega).
\]
The bound is unchanged by~\eqref{eq:fp-weight-closure} at every weight
if and only if
\begin{equation}
 \left|2f(x)-h(x)-2f(y)+h(y)\right|
 \leq\mathsf d_q(p_x,p_y)
 \qquad(0\leq x,y\leq1).
 \label{eq:fp-weight-lipschitz}
\end{equation}

For $q>2$, an equivalent explicit criterion consists of the following
two sets of conditions.

\smallskip
\noindent\textup{(a) Within each branch:}
\begin{align}
 0\leq f(y)-f(x)&\leq h(y)-h(x),
 &&0\leq x\leq y\leq\frac{q-1}{q},
 \label{eq:fp-lower-branch}\\
 0\leq f(x)-f(y)&\leq h(x)-h(y),
 &&\frac{q-1}{q}\leq x\leq y\leq1.
 \label{eq:fp-upper-branch}
\end{align}

\noindent\textup{(b) Between the branches:}
for $0\leq x\leq(q-1)/q\leq y\leq1$, put
\[
 a_y=(q-1)(1-y),\qquad b_y=1-\frac{y}{q-1}.
\]
Then
\begin{align}
 0\leq f(y)-f(x)&\leq h(y)-h(x),&&x\leq a_y,
 \label{eq:fp-cross-left}\\
 0\leq f(x)-f(y)&\leq h(x)-h(y),&&x\geq b_y.
 \label{eq:fp-cross-right}
\end{align}
For $a_y<x<b_y$, define as in~\eqref{eq:high-weight-join}
\[
 s_{x,y}=\left(1-x,
 \underbrace{\frac{x+y-1}{q-2},\ldots,
 \frac{x+y-1}{q-2}}_{q-2\text{ entries}},1-y\right).
\]
The remaining condition is
\begin{equation}
 H_q(s_{x,y})-h(y)
 \leq f(x)-f(y)
 \leq h(x)-H_q(s_{x,y}).
 \label{eq:fp-cross-incomparable}
\end{equation}
\end{theorem}

\begin{proof}
As in~\eqref{eq:fp-directed}, unimprovability is equivalent to
\begin{equation}
 f(x)-f(y)\leq K_q(p_x,p_y)
 \qquad(0\leq x,y\leq1).
 \label{eq:fp-weight-directed}
\end{equation}
Combining the inequalities in both directions gives
\eqref{eq:fp-weight-lipschitz}.

On each side of $(q-1)/q$, the symmetric compositions form a
majorization chain. Specifically,
\[
 p_y\preceq p_x: \quad 0\leq x\leq y\leq\frac{q-1}{q},
 \qquad
 p_x\preceq p_y: \quad \frac{q-1}{q}\leq x\leq y\leq1.
\]
The comparable-pair costs therefore give
\eqref{eq:fp-lower-branch}--\eqref{eq:fp-upper-branch}.

Suppose $q>2$ and $x\leq(q-1)/q\leq y$.
The partial sums of the sorted compositions, for $1\leq k\leq q-1$,
are
\[
 P_k=1-x+\frac{(k-1)x}{q-1},\qquad
 Q_k=\frac{ky}{q-1}.
\]
Their difference is affine in $k$. Comparing its values at $k=1$
and $k=q-1$ gives
\[
 p_y\preceq p_x\iff x\leq a_y,
 \qquad
 p_x\preceq p_y\iff x\geq b_y.
\]
These cases yield~\eqref{eq:fp-cross-left} and
\eqref{eq:fp-cross-right}. In the remaining case $a_y<x<b_y$, we have $p_x^\downarrow\vee p_y^\downarrow=s_{x,y}$ by Proposition~\ref{prop:complete-high-weight}.

The two directions of~\eqref{eq:fp-weight-directed} now give
\eqref{eq:fp-cross-incomparable}. At $x=a_y$ or $x=b_y$ the join
reduces to the appropriate comparable composition. All pairs have
been covered, proving necessity and sufficiency of (a)--(b).
\end{proof}

Conditions~\eqref{eq:fp-lower-branch}--\eqref{eq:fp-upper-branch}
say that both $f$ and its entropy deficit $h_q-f$ are nondecreasing
up to $(q-1)/q$ and nonincreasing thereafter. For $q>2$, the full
criterion also includes~\eqref{eq:fp-cross-left}--\eqref{eq:fp-cross-incomparable}:
the symmetric one-parameter family is not closed under taking joins.

\begin{corollary}[Structural properties of the exact asymptotic rates]
\label{cor:fp-exact-rates}
For every fixed $q\geq2$ and $\delta\in[0,1]$, the following hold.
\begin{enumerate}
\item The exact constant-composition rate is a fixed point: $(\Pp R_q)(\delta;p)=R_q(\delta;p)$.
Both $R_q(\delta;\cdot)$ and $H_q-R_q(\delta;\cdot)$ are
Schur-concave. In particular,
\begin{equation}
 0\leq R_q(\delta;p)-R_q(\delta;r)
 \leq H_q(p)-H_q(r)
 \qquad(p\preceq r).
 \label{eq:fp-rate-comparable}
\end{equation}
\item The exact constant-weight rate is also unchanged by weight
transport:
\[
 R_q(\delta,\omega)
 =\inf_{0\leq\theta\leq1}
 \{K_q(p_\omega,p_\theta)+R_q(\delta,\theta)\}.
\]
It satisfies all the conditions of Theorem~\ref{thm:fp-weight}.
In particular, both functions $\omega\longmapsto R_q(\delta,\omega)$ and $\omega\longmapsto h_q(\omega)-R_q(\delta,\omega)$
are nondecreasing on $[0,(q-1)/q]$ and nonincreasing on $[(q-1)/q,1]$.
\item The function $f(\omega)=R_q(\delta,\omega)$ is absolutely
continuous on $[0,1]$, and almost everywhere
\begin{equation}
 \begin{cases}
 0\leq f'(\omega)\leq
 \displaystyle\log_q\frac{(q-1)(1-\omega)}{\omega},
 &0<\omega<\frac{q-1}{q},\\[2mm]
 \displaystyle\log_q\frac{(q-1)(1-\omega)}{\omega}
 \leq f'(\omega)\leq0,
 &\frac{q-1}{q}<\omega<1.
 \end{cases}
 \label{eq:fp-rate-derivatives}
\end{equation}
\item With $u_q=(1/q,\ldots,1/q)$,
\begin{equation}
 0\leq R_q(\delta;u_q)-R_q(\delta;p)\leq1-H_q(p).
 \label{eq:fp-rate-uniform}
\end{equation}
Consequently,
\begin{equation}
 0\leq f \left (\frac{q-1}{q} \right) - f(\omega)\leq1-h_q(\omega),
 \qquad f'\left(\frac{q-1}{q} \right)=0.
 \label{eq:fp-rate-center}
\end{equation}
The derivative at $(q-1)/q$ exists without any additional regularity
assumption.
\end{enumerate}
\end{corollary}

\begin{proof}
Corollary~\ref{cor:asymptotic-permutation-transport} and
Theorem~\ref{thm:cc-join-cost} give directly
\[
 R_q(\delta;p)\leq K_q(p,r)+R_q(\delta;r)
 \qquad(p,r\in\Delta_{q-1}).
\]
Taking the infimum over $r$, and using $r=p$ for the reverse
inequality, proves the first fixed-point identity. Theorem~\ref{thm:fp-composition}
then gives the Schur-concavity assertions and~\eqref{eq:fp-rate-comparable}.
Restricting to $p=p_\omega$, $r=p_\theta$, and applying
Proposition~\ref{prop:uniform-composition} proves the second fixed-point
identity. Theorem~\ref{thm:fp-weight} gives the stated monotonicities
and all cross-branch inequalities.

On either closed branch, the increment inequalities imply
\[
 |f(y)-f(x)|\leq|h_q(y)-h_q(x)|.
\]
Since $h_q$ is absolutely continuous and monotone on each branch,
the definition of absolute continuity, applied to disjoint intervals,
shows that $f$ is absolutely continuous there as well. The two
restrictions agree at $(q-1)/q$, so $f$ is absolutely continuous on
$[0,1]$. Differentiating the increment inequalities at points where
$f'$ exists gives~\eqref{eq:fp-rate-derivatives}, since
\[
 h_q'(\omega)=\log_q\frac{(q-1)(1-\omega)}{\omega}.
\]

Finally, $u_q\preceq p$ for every $p$, so
\eqref{eq:fp-rate-comparable} implies~\eqref{eq:fp-rate-uniform}.
Taking $p=p_\omega$ gives the inequality in~\eqref{eq:fp-rate-center}.
The Taylor expansion
\[
 1-h_q(\omega)
 =\frac{q^2}{2(q-1)\ln q} \left (\omega-\frac{q-1}{q} \right)^2
 +O \left (\left |\omega-\frac{q-1}{q} \right |^3 \right)
\]
then proves $f'((q-1)/q)=0$.
\end{proof}

\paragraph{The binary case.}
For $q=2$, the complete criterion for $f = B_2(\delta,\omega)$ instead takes the form
\begin{equation}
 f(\omega)=F(h_2(\omega)),\qquad 0\leq\omega\leq1,
 \label{eq:fp-binary}
\end{equation}
where $F:[0,1]\to\mathbb R$ is nondecreasing and $1$-Lipschitz.

Indeed, $p_\omega$ and $p_{1-\omega}$ differ only by a coordinate
permutation. Their costs in both directions are zero, so
$f(\omega)=f(1-\omega)$ is necessary. The compositions on
$[0,1/2]$ form a chain, and $h_2$ maps that interval continuously and
strictly increasingly onto $[0,1]$. Its inverse defines $F$ in
\eqref{eq:fp-binary}; the lower-branch increment inequalities are
exactly
\[
 0\leq F(t)-F(s)\leq t-s\qquad(0\leq s\leq t\leq1).
\]
Conversely, these inequalities and the symmetry implied by
\eqref{eq:fp-binary} give~\eqref{eq:fp-weight-directed} for every pair.

The monotonicity of $R_2(\delta,\omega)$ and of $h_2(\omega)-R_2(\delta,\omega)$ on $[0,1/2]$ is classical.
Both assertions are stated explicitly in Samorodnitsky~\cite[Section~4, Lemmas~4.1--4.2]{Samorodnitsky2001}, with attribution
to Levenshtein's work.

\section{Base bounds to which the closure operator is applied}
\label{sec:base-bounds}

This section describes weight-dependent and unrestricted asymptotic upper bounds.
We shall collect bounds $B_q (\delta,\theta)$, suitably abbreviated, which will serve as input for the closure operator $\Pp$.
The purpose of the section is broader than the preparation to the final
numerical comparison. A bound is omitted from the latter envelope whenever it is contained pointwise in a kept family, or when
no closed proportional-parameter evaluation is presently available.

\subsection{Shell size, shell balls, and Plotkin--Johnson}

The shell $J_q(n,w)$ (see~\eqref{eq:defShell}) itself has exponent $h_q(\omega)$.
The asymptotic form of the unrestricted Singleton bound gives the additional elementary ceiling
$ R_q(\delta,\omega)\leq 1-\delta$.
Neither the shell-size ceiling nor the Singleton ceiling
is used in the final comparison: on its domain each is no stronger than the
kept families.

For $q\geq3$, let $\alpha,\beta\geq0$ satisfy $\alpha\leq1-\omega$, $\alpha+\beta\leq\omega$,
and put $\gamma=\omega-\alpha-\beta$. Define
\[
 \Phi_q^{\mathrm{rel}}(\omega;\alpha,\beta) := (1-\omega)h_q\left(\frac{\alpha}{1-\omega}\right) +\omega H_q \left(\frac{\alpha}{\omega},
\frac{\beta}{\omega}, \frac{\gamma}{\omega}\right)\notag
+ \beta\log_q(q-2).
\]
This is the exponent of the $(\alpha,\beta)$-relation in $\mathcal{A}^n$ with respect to
a fixed
word of weight $\omega n$. Hence the shell-ball exponent is,
\begin{equation} \label{eq:Ball}
 V_q(\rho,\omega) = \max_{\mathcal{D}} \Phi_q^{\mathrm{rel}}(\omega;\alpha,\beta), 
\end{equation}
where the optimization domain $\mathcal{D}$ contains $\alpha,\beta \geq 0$ such that $\alpha\leq1-\omega$, $\alpha+\beta\leq\omega$, and $2\alpha+\beta\leq\rho$.
Then the shell Hamming bound is
\[
 B_q^{\mathrm{Ham}}(\delta,\omega) = h_q(\omega)-V_q(\delta/2,\omega).
\]
For $q=2$ the corresponding formulas are obtained from the binary
Johnson scheme and should be written separately rather than by substituting
$q=2$ into $\Phi_q^{\mathrm{rel}}(\omega;\alpha,\beta)$.

The above Hamming--Johnson construction belongs to the relation-layer Bassalygo--Elias bound envelope, which we consider below, so
adding the shell-Hamming curve cannot lower the selected envelope.

The asymptotic average-distance threshold~\cite{Johnson1962,Levenshtein1971} in the symmetric
constant-weight shell is
\[
\delta_q^{\mathrm{Pl}}(\omega) = 1-(1-\omega)^2-\frac{\omega^2}{q-1} =
 2\omega-\frac{q}{q-1}\omega^2.
\]
Therefore
\begin{equation}
B^{\mathrm{Pl}}_q(\delta,\omega) = R_q(\delta,\omega)=0
 \qquad\text{for}\qquad
 \delta\geq\delta_q^{\mathrm{Pl}}(\omega).
 \label{eq:plotkin-zero}
\end{equation}
The zero-rate region of this Plotkin bound is already enforced by the following relation-layer Bassalygo–Elias bound.

\subsection{Relation-layer Bassalygo--Elias bound}

For $q\geq3$, define
\[
 \delta_q^{\mathrm{loc}}(\omega;\alpha,\beta) := 2\alpha-
 \frac{q\alpha^2}{(q-1)(1-\omega)} + \omega -\frac{\alpha^2+\gamma^2+\beta^2/(q-2)}{\omega},
 \qquad \gamma=\omega-\alpha-\beta.
\]
Averaging over a relation layer and applying Plotkin's bound inside the
selected subcode gives
\begin{equation}
 B_q^{\mathrm{rel-BE}}(\delta,\omega) =
 h_q(\omega) - \sup_{\substack{\alpha,\beta\geq0\\
 \alpha\leq1-\omega,\ \alpha+\beta\leq\omega\\
 \delta_q^{\mathrm{loc}}(\omega;\alpha,\beta)\leq\delta}}
 \Phi_q^{\mathrm{rel}}(\omega;\alpha,\beta).
 \label{eq:relation-eb}
\end{equation}

This family gives best-known bounds for some values of parameters and so we keep it in the final envelope.

\subsection{Code--anticode bounds}
\label{sec:code--anticode}

The code--anticode principle provides a second general way of converting
diametric information inside a shell $J_q(n,w)$ into an upper bound.
Let
\[
I_q(n, D, w) := \max \{|A| : A \subset J_q(n, w), \diam (A) \leq D \}.
\]	

For a vertex transitive graph $G = (V,E)$ the folklore clique-coclique bound (see, for instance,~\cite{GM}) gives
\begin{equation} \label{cliq-co}
 \alpha(G) \cdot \omega(G) \leq |V(G)|.
\end{equation}

Now let $G$ be a graph with the vertex set $J_q(n, w)$ and whose edges connect pairs of vertices at the Hamming distance at most $d-1$. Clearly $G$ is vertex transitive, so~\eqref{cliq-co} gives
\begin{equation} \label{eq:finite-code--anticode}
 A_q(n,d,w) \cdot I_q(n,d-1,w) \leq |J_q(n,w)| = \binom nw(q-1)^w.
\end{equation}
In particular, every explicit anticode gives a valid upper bound, whether
or not it is known to be maximum.

\paragraph{Binary constant-weight codes.}
For $q=2$, write $s:=\left\lfloor(d-1)/2\right\rfloor$ and $t:=w-s$.
An anticode of Hamming diameter at most $2s$ is the same as a
$t$-intersecting family of $w$-subsets. The complete intersection theorem
of Ahlswede and Khachatrian~\cite{AhlswedeKhachatrian1997} determines its
maximum cardinality:
\begin{equation}
 I_2(n,2s,w)
 =
 \max_i
 \sum_{j=t+i}^{\min\{w,t+2i\}}
 \binom{t+2i}{j}
 \binom{n-t-2i}{w-j},
 \label{eq:AK-anticode-finite}
\end{equation}
where
\[
 0\leq i\leq \min\left\{w-t,\left\lfloor\frac{n-t}{2}\right\rfloor\right\}, \qquad w \leq n/2, \qquad 0 \leq s \leq w.
\]
Consequently
\[
A_2(n,d,w) \leq \frac{\binom nw}{I_2(n,2\lfloor(d-1)/2\rfloor,w)}.
\]

For numerical asymptotics one may apply Stirling's formula directly to
\eqref{eq:AK-anticode-finite}. Assuming $0\leq\omega\leq1/2$ (the other
half follows by complementation) and $0\leq\delta\leq2\omega$, put
$\tau_0:=\omega-\delta/2$.
Then the anticode exponent is
\[
 \alpha_2^{\mathrm{AK}}(\delta,\omega)
 :=
 \sup_{\mathcal{D}} \left ( (\tau_0+2\iota)
 h_2\!\left(\frac{x}{\tau_0+2\iota}\right)\notag +
 (1-\tau_0-2\iota)
 h_2\!\left(
 \frac{\omega-x}{1-\tau_0-2\iota}
 \right) \right),
\]
where the optimization is over $(x, \iota) \in \mathcal{D}$ defined by
$0\leq\iota\leq \min\{\delta/2,(1-\tau_0)/2\}$, $x\geq\tau_0+\iota$, $x\leq\min\{\omega,\tau_0+2\iota\}$ and 
$\omega - x \leq 1-\tau_0-2\iota$.
Hence the complete-intersection code--anticode bound in the binary case is
\[
 B_2^{\mathrm{CA}}(\delta,\omega) = h_2(\omega)-\alpha_2^{\mathrm{AK}}(\delta,\omega).
\]

This binary complete-intersection curve is recorded here because of its theoretical importance. 
It is covered by the moving-subspace Johnson-layer construction appearing below.

\paragraph{Nonbinary constant-weight codes.}
For $q>2$, the code--anticode inequality
\eqref{eq:finite-code--anticode} is available, but the maximum anticode
size is not known for all proportional parameters. Etzion~\cite{Etzion2021} identifies
several families of maximum anticodes for special parameter ranges and
leaves the general diametric problem open. Two explicit
families were nevertheless used asymptotically.

First, fix the support of the word and allow $D$ specified nonzero
coordinates to vary over $\mathcal A\setminus\{0\}$, while fixing the
remaining nonzero symbols. This gives an anticode of diameter at most $D$ and size $(q-1)^D$. Second, fix $\tau n$ coordinates to one prescribed nonzero symbol and
allow the remaining $(1-\tau)n$ coordinates to run through $J_q((1-\tau)n,(\omega-\tau)n)$. Its exponent is
\[
 a_q^{\mathrm{s}}(\tau,\omega) = (1-\tau) h_q\!\left(\frac{\omega-\tau}{1-\tau}\right)
\]
and its Hamming diameter has relative value at most
$d_q^{\mathrm{s}}(\tau,\omega) = \min\{2(\omega-\tau),\,1-\tau\}$.
Thus an available anticode exponent is 
\[
 \alpha_q^{\mathrm{ex}}(\delta,\omega) := \max\left\{
 \min\{\delta,\omega\}\log_q(q-1),\
 \sup_{\substack{0\leq\tau\leq\omega\\
 d_q^{\mathrm{s}}(\tau,\omega) < \delta}}
 a_q^{\mathrm{s}}(\tau,\omega) \right\}
\]
yielding the bound
\[
 B_q^{\mathrm{CA,ex}}(\delta,\omega) := h_q(\omega)-\alpha_q^{\mathrm{ex}}(\delta,\omega),
 \qquad q>2.
\]
The superscript ``ex'' here means that the bound follows from a particular anticode example: unlike the binary
complete-intersection expression where the diametric problem has been solved.

The explicit nonbinary code--anticode bound is covered by the $B_q^{cc-CA}$ bound.

For the unrestricted Hamming space, the maximum-anticode problem is known
exactly by the diametric theorem of Ahlswede and Khachatrian~\cite{AhlswedeKhachatrian1998}. 
This underlies, in particular, the optimal-anticode refinements used in nonbinary unrestricted bounds such
as Kaipa's~\cite{Kaipa2018}.

\subsection{The Ben-Haim--Litsyn product-scheme bound}
\label{sec:bhl}

We use the finite-dimensional asymptotic formulation of Ben-Haim and
Litsyn~\cite{BenHaimLitsyn2007}. As written below, the formulas require
$q\geq3$; boundary values are understood by continuity. Define
\[
 k_q(x) = \frac{q-1}{q} - \frac{q-2}{q}x - \frac2q\sqrt{(q-1)x(1-x)},
 \qquad
 h(x,y) = \frac{x(1-x)-y(1-y)}{1+2\sqrt{y(1-y)}}.
\]
Put
\[ 
R_1(q,\omega,\delta):= \min_{\mathcal D_1} 
\biggl\{ h_q(\xi) +\xi h_q\!\left(\frac\eta\xi\right) -(\xi+\eta)\log_q(q-1) +\eta\log_q(q-2) \biggr\},
\]
where the optimization domain $\mathcal D_1 \ni (\eta,\xi)$ is defined by inequalities
\[
0 \leq \eta \leq \frac{q-2}{q-1}\omega, \qquad 0 \leq\xi-\eta\leq\min\{\omega-\eta,1-\omega\},
\]
\[
b = (1-\eta) h\!\left( \frac{\omega-\eta}{1-\eta}, \frac{\xi-\eta}{1-\eta} \right) \leq \omega-\frac{q-1}{q-2}\eta, \qquad
\delta \geq 2b+(\omega-b) k_{q-1}\!\left(\frac{\eta}{\omega-b}\right).
\]
Next define
\[
 R_2(q,\omega,\nu,\phi,\delta) := \min_{\mathcal{D}_2}
 \biggl\{
 \nu R_1\!\left(q,\frac\phi\nu,\delta_1\right) + (1-\nu) R_1\!\left( q,\frac{\omega-\phi}{1-\nu},\delta_2 \right)
 \biggr\},
\]
where the domain $\mathcal D_2 \ni (\delta_1,\delta_2)$ is defined by $\nu\delta_1+(1-\nu)\delta_2=\delta$, $0\leq\delta_1\leq\min\{2\phi/\nu,1\}$,
$0\leq\delta_2\leq \min\{2(\omega-\phi)/(1-\nu),1\}$.
Finally, the Ben-Haim--Litsyn upper bound on $R_q(\delta,\omega)$ is
\[
 B_q^{\BHL}(\delta,\omega) :=
 \min_{\substack{ 0\leq\phi\leq\nu\\
 0\leq\omega-\phi\leq1-\nu}}
 \biggl\{ h_q(\omega) -\nu h_q\!\left(\frac\phi\nu\right) -(1-\nu) h_q\!\left(\frac{\omega-\phi}{1-\nu}\right)
 + R_2(q,\omega,\nu,\phi,\delta) \biggr\}.
\]

\subsection{The nonbinary Johnson LP and finite SDP bounds}

For $q>2$, two words in $J_q(n,w)$ are classified by two orbit
parameters: $j$ support deletions/insertions and $i$ nonzero symbol
changes on the common support. Their Hamming distance is $i+2j$, and
the relation valency is
\[
v_{ij} = (q-1)^j \binom wj \binom{n-w}{j} \binom{w-j}{i} (q-2)^i.
\]
Let $Q_{rs}(i,j)$ denote the second eigenmatrix of the nonbinary Johnson
scheme. The exact finite Delsarte linear programming (LP) relaxation is
\begin{align}
 \text{maximize }&
 \sum_{i,j}a_{ij},\notag\\
 \text{subject to }&
 a_{00}=1,\qquad
 a_{ij}\geq0,\qquad
 a_{ij}=0\quad(0<i+2j<d),\notag\\
 &
 \sum_{i,j}a_{ij}Q_{rs}(i,j)\geq0
 \quad\text{for every dual index }(r,s).
\end{align}
Its normalized asymptotic optimum, denoted
$B_{q,\LP}^{\CW}(\delta,\omega)$, is a valid upper bound and contains
every explicit feasible dual polynomial in this scheme~\cite{Delsarte1973,TarnanenAaltonenGoethals1985}. 
Aaltonen's constructions and the Ben-Haim--Litsyn product Hamming--Johnson
construction are low-dimensional explicit or numerical realizations
within this broader LP framework. Terwilliger-algebra and higher-point
SDP refinements~\cite{Schrijver2005,GijswijtSchrijverTanaka2006,Polak2019}
can improve finite-length bounds, but no closed proportional-parameter
formula comparable to the binary results below is presently available.
Consequently nonbinary Johnson LP bound and its SDP refinements are not plotted.

\subsection{Unrestricted nonbinary ceilings}
\label{sec:unrestricted}

Every constant-weight code is an unrestricted code. Hence any
unrestricted upper bound $U_q(\delta)$ gives
\[
 R_q(\delta,\omega)\leq U_q(\delta).
\]
The first $q$-ary Delsarte--Levenshtein curve is
\begin{equation}
 U_q^{\LP1}(\delta)
 =
 h_q\!\left(
 \frac{
 q-1-(q-2)\delta
 -2\sqrt{(q-1)\delta(1-\delta)}
 }{q}
 \right).
 \label{eq:qary-LP1}
\end{equation}

Further published nonbinary unrestricted families include Aaltonen's LP
bound~\cite{Aaltonen1990}, the shortening and straight-line closure of
Laihonen--Litsyn~\cite{LaihonenLitsyn1998}, Kaipa's optimal-anticode
hybrid~\cite{Kaipa2018,AhlswedeKhachatrian1998}, and the unrestricted
Ben-Haim--Litsyn product-scheme bound~\cite{BenHaimLitsyn2007}. The unrestricted BHL bound is defined by
\[
 U_q^{\BHL,\mathrm{unr}}(\delta)
 :=\inf_{0\leq\omega\leq1}
 \{1-h_q(\omega)+B_q^{\BHL}(\delta,\omega)\}.
\]
This is the outer Bassalygo--Elias conversion formalized in Section~\ref{sec:master-BE}.

For comparison, we record Kaipa's Elias--Plotkin
hybrid~\cite{Kaipa2018}. For $q>2$, put
$\vartheta=(q-1)/q$ and
$\delta_1=(2q-3)/(q(q-1))$. Then
\begin{equation}
 U_q^{\mathrm{Kaipa}}(\delta)=
 \begin{cases}
 1-h_q\!\left(\vartheta-
 \sqrt{\vartheta(\vartheta-\delta)}\right),
 &0\leq\delta\leq\delta_1,\\
 (\vartheta-\delta)\dfrac{q-1}{q-2}\log_q(q-1),
 &\delta_1\leq\delta\leq\vartheta,\\
 0,&\vartheta\leq\delta\leq1.
 \end{cases}
 \label{eq:kaipa-elias-plotkin}
\end{equation}

This curve is not a separate component of the
selected envelope, since it is pointwise dominated by the unrestricted
Ben-Haim--Litsyn bound $U_q^{\BHL,\mathrm{unr}}(\delta)$.
Indeed, the choice $\eta=\xi=0$ in the definition of $R_1$ gives
$R_1(q,w,\delta_q^{\mathrm{Pl}}(w))=0$.
For $0\leq\delta\leq\delta_1$, choosing
\[
 w=\vartheta-\sqrt{\vartheta(\vartheta-\delta)}
\]
and applying outer Bassalygo--Elias averaging recovers the first
branch of~\eqref{eq:kaipa-elias-plotkin}.
For $\delta_1\leq\delta\leq\vartheta$, use the two-block BHL construction
with block weights $1/q$ and $\vartheta$, block distances $\delta_1$
and $\vartheta$, and first-block proportion $\nu = (\vartheta-\delta)/(\vartheta-\delta_1)$.

Both block bounds vanish, and outer averaging yields $\nu[1-h_q(1/q)] = (\vartheta-\delta)\frac{q-1}{q-2}\log_q(q-1)$.
Thus $U_q^{\BHL,\mathrm{unr}}(\delta) \leq U_q^{\mathrm{Kaipa}}(\delta)$.
This ordering is preserved by the generalized-shortening closure
defined below, because that operator is monotone in its input.

Aaltonen bound is contained in the BHL family, while LP1 bound~\eqref{eq:qary-LP1} is contained in the mixed-channel family defined below. The Laihonen--Litsyn shortening construction is incorporated through $\mathcal S$, so we do not consider it as a separate seed.

For an unrestricted upper bound $U$, define its
generalized-shortening closure by
\[
 (\mathcal S U)(\delta)
 =\inf
 \left\{
 \alpha\bigl[1-\widehat h_q(\rho/\alpha)\bigr]
 +(1-\alpha)U(\Delta):
 \rho=\frac{\delta-(1-\alpha)\Delta}{2},\
 0\leq\alpha,\Delta\leq1,\ 0\leq\rho\leq\alpha
 \right\},
\]
where $\widehat h_q(x) := h_q(\min\{x,(q-1)/q\})$ and the face $\alpha=0$ is
understood by continuity. This is the asymptotic form of the
Laihonen--Litsyn shortening inequality~\cite{LaihonenLitsyn1998}: $\rho$
is the normalized radius in the deleted block and $\Delta$ is the relative
distance of the residual code.

Alrabiah and Guruswami give a $q$-ary extension of their
classical--quantum channel criterion and an explicit $q$-ary mixed-channel
family~\cite{AlrabiahGuruswami2026}. With
\begin{align*}
 \nu_q(a)&=1-\frac{q}{q-1}a,
 &
 r_q(v)&=\frac{q-1}{q}(1-v),\\
 \gamma_q(x)&=
 \frac{
 q-1-(q-2)x
 -2\sqrt{(q-1)x(1-x)}
 }{q},
\end{align*}
put, for $v\in[\nu_q(\delta),1]$,
\[
 t=\sqrt{\frac{\nu_q(\delta)}{v}},
 \qquad
 g=\gamma_q(r_q(v)),
 \qquad
 \ell=\frac{(1-t)g}{q-1},
\]
and
\[
 \lambda_{1,2}
 =
 \frac12
 \left(
 1-(q-2)\ell
 \pm
 \sqrt{
 (1-(q-2)\ell)^2
 -
 \frac{
 4(1-t)(1+(q-1)t)
 }{q-1}
 g(1-g)
 }
 \right).
\]
The resulting $q$-ary mixed-channel ceiling for $\delta < \frac{q-1}{q}$ is
\[
 U_q^{\mathrm{MC}}(\delta) = \inf_{\nu_q(\delta)\leq v\leq1}
 \left\{ h_q(g) +\lambda_1\log_q\lambda_1 +\lambda_2\log_q\lambda_2 +(q-2)\ell\log_q\ell \right\}.
\]
For $q>2$, the source~\cite{AlrabiahGuruswami2026} does not claim that this curve dominates the
Aaltonen, Laihonen--Litsyn, or Ben-Haim--Litsyn envelopes, so these
families must be compared pointwise.

In the final comparison, the mixed-channel and unrestricted Ben-Haim–Litsyn bounds are each replaced by their
generalized-shortening closure. The unshortened curves are redundant because zero shortening recovers each seed.

\subsection{Recent binary moving-subspace and LP developments}
\label{sec:recent-binary}

The strongest recent developments are currently specific to the binary
Hamming and Johnson schemes. They are therefore relevant both as
additional binary base bounds and as structural guidance, but they should
not be inserted into the $q>2$ envelope without a separate nonbinary
argument.

Chapter~2 of \textit{Ten Advances in Mathematics and Theoretical Computer
Science}~\cite{OpenAI2026} replaces the single stabilizer-fixed vector in
a Delsarte eigenspace by a moving higher-rank stabilizer module. In the
binary Hamming cube, define
\[
 \Gamma_H(a,b) = \frac{2(a-b)(1-a-b)}{\sqrt{a(1-a)}}
\]
and
\[
 \kappa_H(\delta) = \inf_{\substack{ 0\leq b<a\leq1/2\\ \Gamma_H(a,b)>1-2\delta}} \{h_2(a)-h_2(b)\}.
\]
Barg~\cite{Barg2026} shows explicitly that this Hamming cube
moving-subspace expression is identical to the mixed-qubit
classical--quantum bound of Alrabiah and Guruswami:
\[
 \kappa_H(\delta) = R_{\mathrm{MQC}}(\delta).
\]
Thus the two apparently different arguments give the same first new
binary curve. Barg~\cite{Barg2026} also explains the analogous relationship for the
recent improvements of the second MRRW bound: one route works on
constant-weight Johnson layers and then applies Bassalygo--Elias
averaging, while the classical--quantum route uses masking. In both
pictures a subspace is attached to each codeword and moves with it.

The moving-subspace construction also has a genuinely weight-sensitive
binary branch, like in the classical LP case. For
\[
 \frac{\delta}{2}<\alpha<\frac12,\qquad
 0\leq\beta<\frac{\alpha}{2},\qquad
 0\leq\gamma<\frac{1-\alpha}{2},
\]
and $\beta+\gamma<u< \min\{ \alpha,\, \alpha-\beta+\gamma,\,
 1-\alpha+\beta-\gamma \}$,
put $z=1-2u$, $m=1-2\alpha$, $\zeta=1-2\beta-2\gamma$, $\xi=1-2\alpha+2\beta-2\gamma$, 
and
\[
 \Lambda_{\alpha,\beta,\gamma}(u) = \frac{(\zeta\xi-mz^2)^2} {z^2(1-m^2)(1-z^2)}
 + \frac{(z^2-\xi^2)(\zeta^2-z^2)} {z^2(1-m^2)\sqrt{1-z^2}}.
\]
The corresponding binary fixed-weight bound is
\[
 B_2^{\mathrm{MS-CW}}(\delta,\alpha) :=
 \inf \left( h_2(u) -\alpha h_2\!\left(\frac{\beta}{\alpha}\right) -(1-\alpha)
 h_2 \!\left(\frac{\gamma}{1-\alpha}\right) \right),
\]
where infimum is taken over $\Lambda_{\alpha,\beta,\gamma}(u) > 1-\frac{\delta}{2\alpha(1-\alpha)}$.
Applying outer Bassalygo--Elias averaging to this layer gives the
corresponding unrestricted improvement of the second MRRW bound.

Write $U_2^{\mathrm{HMS}}$ for the Hamming-cube
moving-subspace curve defined by $\kappa_H$ above and set
\begin{equation}
 U_2^{\mathrm{MS}}(\delta)=
 \min\left\{U_2^{\mathrm{HMS}}(\delta),
 \inf_{0<\alpha<1/2}
 \bigl[1-h_2(\alpha)+B_2^{\mathrm{MS-CW}}(\delta,\alpha)\bigr]
 \right\}.
 \label{eq:binary-combined-ms}
\end{equation}
This combined family is the binary unrestricted
moving-subspace input used below. The Hamming moving-subspace and binary
mixed-qubit channel formulas coincide and hence are not drawn as separate
curves; the masked mixed-qubit curve is also omitted because
\eqref{eq:binary-combined-ms} is pointwise no worse. Similarly, the
one-vector binary Johnson LP constant-weight certificate is not kept,
because the moving-subspace Johnson-layer construction strengthens it. 
The full permutation transport contains the untransported $B_2^{\mathrm{MS-CW}}$ curve.

Gay, Jeronimo, and Liu's Honeycomb framework enlarges the Hamming cube
representation graph and gives a first-level exponent no worse than the
combined moving-subspace exponent, with strict improvement whenever the
new branch is active~\cite{GayJeronimoLiu2026}. Barg~\cite{Barg2026} provides a
coding-theoretic explanation of the common moving-subspace mechanism
behind these constructions.

A further development appeared in Salmon~\cite{Salmon2026}. Let
$LP_n(d)$ denote the ordinary radial binary Delsarte LP optimum and
\[
R_D(\delta) := \limsup_{n\to\infty} \frac{1}{n}\log_2 LP_n(\lceil\delta n\rceil).
\]
Salmon proves that the half-rate point of the \textit{LP relaxation itself}
is exact. With
\[
 \delta_\pi := \frac12-\frac1\pi,
\]
one has
\[
 R_D(\delta_\pi)=\frac12, \qquad R_D(\delta)<\frac12 \quad \text{for} \quad \delta>\delta_\pi.
\]
Moreover, the construction in the same paper gives the explicit
Delsarte ceiling
\[
 R_2^{\mathrm{unr}}(\delta) \leq R_D(\delta) \leq U_2^{\mathrm{Sal}}(\delta) := \frac{\pi}{4}(1-2\delta),
 \qquad
 \delta_\pi\leq\delta\leq\frac12.
\]
At the endpoint $\delta_\pi$ this equals $1/2$, so the bound is sharp for
the radial Delsarte relaxation.

Salmon obtains this result by extending the finite moving-subspace method
to higher multiplicity spaces through a generalized Hecke algebra and by
using all levels of the hierarchy simultaneously. 

For orientation, Salmon records the following values at
$\delta=\delta_\pi$:
\[
\begin{array}{l|c}
\text{binary upper-bound construction}
&\text{exponent at }\delta_\pi\\
\hline
\text{first MRRW}&0.5130793807\ldots\\
\text{second MRRW}&0.5017797958\ldots\\
\text{combined moving-subspace bound}&0.5005547387\ldots\\
\text{first Honeycomb level}&0.5003477743\ldots\\
\text{full radial Delsarte LP}&0.5 
\end{array}
\]
Thus the finite-level 2026 improvements should not be regarded as the
endpoint of the binary LP method for unrestricted codes: at least at the half-rate point, the
full radial LP is now evaluated exactly. The numerical gains over the classical asymptotic bounds remain modest.
Their scale is consistent with the numerical evidence of Barg and
Jaffe~\cite{BargJaffe2001}, which suggested limited scope for substantial
improvements within the binary Delsarte LP method. A related perspective
for $q\geq3$ was developed by Boyvalenkov, Danev, and
Stoyanova~\cite[Section~2.3]{BoyvalenkovDanevStoyanova2018}, who conjectured
that their refinements of the Levenshtein bounds attain the full LP
optimum for sufficiently large relative distances.

The two September 2026 papers do not presently supply a
corresponding new $q$-ary constant-weight formula. Barg~\cite{Barg2026} notes that the
spectral method extends in principle to other homogeneous spaces,
including the $q$-ary Hamming scheme, but that implementing the new
moving-subspace refinement there requires additional work. Salmon's half-rate theorem is binary.

Consequently, the nonbinary final comparison keeps only
the non-dominated explicit families identified above: Ben-Haim--Litsyn constant weight, the relation-layer and explicit shell-anticode bounds, 
and the two shortened unrestricted families.

\subsection{Composition-sensitive target bounds}
\label{sec:cc-base-bounds}

Here and in the sequel we will use cc to abbreviate the constant composition case or bound.

\paragraph{Plotkin and type-class packing.}
For a composition $r$, the Plotkin threshold is
\[
 \Delta(r)=1-\sum_i r_i^2.
\]
Indeed, for a code of size $M$, let $N_{i,j}$ count occurrences of symbol
$i$ in coordinate $j$. Its sum of ordered pair distances is
\[
 \sum_j\left(M^2-\sum_i N_{i,j}^2\right)
 \leq nM^2\left(1-\sum_i r_i^2\right),
\]
since $\sum_j N_{i,j}=Mnr_i$. Thus the rate is zero for
$\delta>\Delta(r)$; the endpoint follows by the usual shortening limit.

Define $B_q^{\mathrm{cc-Pl}}(\delta;r)=0$ for $\delta\geq\Delta(r)$ and $B_q^{\mathrm{cc-Pl}}(\delta;r)=H_q(r)$ otherwise. Its full transport is kept; the shell-Plotkin redundancy does not apply to arbitrary compositions.

For the sake of completeness, define the type-class ball exponent
\[
 V_{\mathrm{cc}}(\rho;r) :=\max_{\substack{Q\geq0,\;Q\mathbf1=Q^{\mathsf T}\mathbf1=r\\
 1-\operatorname{tr}Q\leq\rho}}\{H_q(Q)-H_q(r)\}.
\]
Here $H_q(Q)$ is the entropy of the $q^2$ entries. A word with
composition $nr$ has $q^{n(H_q(Q)-H_q(r)+o(1))}$ neighbors of joint type $Q$, and there
are polynomially many joint types. Disjoint balls inside the type class
therefore give
\[
 B^{\mathrm{Ham-cc}}_q(\delta;r) = H_q(r)-V_{\mathrm{cc}}(\delta/2;r).
\]

\paragraph{A constant-composition code--anticode bound.}
Choose $0\leq z_i\leq r_i$, put $\lambda=\sum_i z_i$, fix a prefix
with composition $n(r-z)$, and let the remaining $\lambda n$ positions
run through the type class with symbol counts $nz$. This anticode has
exponent $\lambda H_q(z/\lambda)$, and its diameter is at most
$n\min\{\lambda,2(\lambda-\max_i z_i)\}$.
For the latter claim, positions occupied by a most frequent symbol in two
residual words intersect in at least $n(2\max_i z_i-\lambda)$ positions.
Consequently,
\[
 B_q^{\mathrm{cc-CA}}(\delta;r) := H_q(r)-
 \sup_{\substack{0\leq z_i\leq r_i\\
 \min\{\lambda,2(\lambda-\max_i z_i)\}<\delta}}
\lambda H_q(z/\lambda)
\]
is a valid bound. The value at $\lambda=0$ is interpreted as zero and $B_q^{\mathrm{cc-CA}}(0;r):=H_q(r)$.
The strict diameter constraint is implemented by limits from below.
This is an explicit family, not a characterization of maximum anticodes.
For this particular family the supremum is obtained by a finite union of
capped water-filling problems, implemented in the accompanying code.

After full transport, $\Pp B_q^{\mathrm{cc-CA}}\leq\Pp B_q^{\mathrm{Ham-cc}}$ and $\Pp B_q^{\mathrm{cc-Pl}}\leq\Pp B_q^{\mathrm{Ham-cc}}$, by the rate--distortion reduction in Appendix~A and monotonicity and idempotence of $\Pp$. Thus transported Hamming is omitted; the two remaining composition-sensitive families need not dominate each other, although their necessity in the larger envelope requires a separate comparison.

\subsection{Asymptotic lower bounds}
\label{subsec:asymptotic-lower-bounds}

For the sake of comparison and completeness we briefly recall the principal asymptotic lower bounds (mostly for unrestricted codes). 
We use notation $[x]_+:=\max\{x,0\}$.

\paragraph{Varshamov--Gilbert.}
For every integer $q\geq2$,
\[
 R_q(\delta)\geq1-h_q(\delta),
 \qquad 0\leq\delta\leq1-\frac1q.
\]
For prime power $q$, this bound can be realized by linear
codes~\cite{Gilbert1952,Varshamov1957}. For every fixed $\delta$ Liang~\cite{liang2024q} proposes an improvement of the Varshamov--Gilbert bound 
for all sufficiently large $q$; namely,~\cite{liang2024q} states that for all sufficiently large $q$ and some absolute constant $C > 0$ 
\[
R_q(\delta) \geq 1 - \delta - \frac{C\log \log q}{q^{1/6}}.
\]

\paragraph{Algebraic-geometric bounds.}
For prime-power $q$, let
$A(q):=\limsup_{g\to\infty}N_q(g)/g$ be the Ihara constant (cf.~\cite{ihara1981some,beelen2022survey}), where
$N_q(g)$ is the maximum number of rational points on a smooth
projective geometrically irreducible curve of genus $g$ over
$\mathbb F_q$. Algebraic-geometric codes yield
\[
 R_q(\delta)\geq\left[1-\delta-\frac1{A(q)}\right]_+.
\]
For square prime powers, $A(q)=\sqrt q-1$, giving the
Tsfasman--Vl\u{a}du\c{t}--Zink bound~\cite{TVZ1982}.
Nonlinear algebraic-geometric constructions include the bounds of
Elkies~\cite{elkies2001excellent,elkies2003still}, Xing~\cite{xing2011asymptotically}, Niederreiter--\"Ozbudak~\cite{NiederreiterOzbudak2004,NiederreiterOzbudak2006}, and Liu--Wu--Xing~\cite{LiuWuXing2023}.
One explicit form due to Niederreiter--\"Ozbudak is
\[
 R_q(\delta)\geq
 \left[1-\delta-\frac1{A(q)}
 +\log_q(1+\frac{1}{q^3})\right]_+ . 
\]

\paragraph{Concatenation bounds.}
For prime-power $q$, the Zyablov~\cite{zyablov1971estimate} bound is
\[
 R_q(\delta)\geq
 \max_{\delta\leq t\leq1-1/q}
 (1-h_q(t))\left(1-\frac{\delta}{t}\right),
 \qquad 0<\delta<1-\frac1q.
\]
Multilevel concatenation gives the Blokh--Zyablov family~\cite{blokh1974coding}; see,
for example,~\cite{BargZemor2005}. These bounds are inferior to the Varshamov--Gilbert bound but corresponding codes can be explicitly constructed.

\paragraph{Constant-composition and constant-weight Varshamov--Gilbert bounds.}
For a probability vector $p$, the constant-composition version of the Varshamov--Gilbert bound is
\[
 R_q(\delta;p) \geq 2H_q(p) - \max_{\substack{Q\geq0,\;Q\mathbf1=Q^{\mathsf T}\mathbf1=p\\ \sum_{a\ne b}Q_{ab}\leq\delta}}H_q(Q),
\]
where $H_q(Q)$ is the entropy of the joint distribution
$Q$; see also~\cite{somekh2019generalized}.
In the shell notation introduced above, the constant-weight bound is
\[
 R_q(\delta,\omega)\geq h_q(\omega)-V_q(\delta,\omega),
\]
where $V_q$ is defined in~\eqref{eq:Ball} for $q \geq 3$ and by
\[
 V_2(\rho,\omega)=
 \max_{0\leq a\leq\min\{\omega,1-\omega,\rho/2\}}
 \left\{\omega h_2\!\left(\frac a\omega\right)
 +(1-\omega)h_2\!\left(\frac a{1-\omega}\right)\right\},
\]
for $q = 2$.

\paragraph{Extraction from unrestricted codes.}
Any unrestricted asymptotic lower bound $L_q(\delta)$ also gives,
by averaging over Hamming isometries,
\[
R_q(\delta;p)\geq[L_q(\delta)+H_q(p)-1]_+,\qquad
R_q(\delta,\omega)\geq[L_q(\delta)+h_q(\omega)-1]_+.
\]
These extraction bounds require no prime-power assumption on $q$.

\section{The master optimized envelope}
\label{sec:master-envelope}

The transport of~\eqref{eq:plotkin-zero} is considered later as part of
the optimized closure and is not itself a base bound.

Fix $q \geq 2$ and let $r = (r_0, \dots, r_{q-1})$ be a composition vector. 
Clearly a constant-weight upper bound $B_q$ and an unrestricted upper bound $U_q$ can be extended for $r$ as follows 
\[
 B_q(\delta;r) := B_q(\delta,1-r_0),
 \qquad
 U_q(\delta;r) := U_q(\delta).
\]
An input is assigned $+\infty$ outside its stated domain. 
Now for every upper bound $B_q(\delta;r)$ Theorem~\ref{thm:cc-join-cost} gives
 \[
 R_q(\delta,\omega) \leq h_q(\omega)+\inf_{r\in\Delta_{q-1}}
 \left\{B_q(\delta;r) -H_q(p_\omega^\downarrow\vee r^\downarrow)\right\}.
 \]

Below we consider only the bounds that may give the strongest refinement after applying the transport operator $\Pp$.
A necessary condition for this is to be a best known base bound for some parameter values.

Namely, the list of base bounds is described as follows. For the numerical comparisons, let $U_2^*(\delta)$ denote
the pointwise minimum of the selected binary unrestricted curves (the
combined moving-subspace, Honeycomb and Salmon branches). For
$q>2$, set
\[
 U_q^*(\delta)
 =\min\bigl\{ (\mathcal S U_q^{\BHL,\mathrm{unr}})(\delta),
 (\mathcal S U_q^{\mathrm{MC}})(\delta)\bigr\}.
\]
For $q>2$, define
\[
 B_q^{\mathrm{pure}}(\delta,\omega)
 = \min\{B_q^{\mathrm{rel-BE}}(\delta,\omega),
 B_q^{\mathrm{BHL}}(\delta,\omega)\},
 \]
 \[
 B_q^{\mathrm{cc}+\Pp}(\delta,\omega)
 =\min\{(\Pp B_q^{\mathrm{cc-Pl}})(\delta;p_\omega),
 (\Pp B_q^{\mathrm{cc-CA}})(\delta;p_\omega)\},
\]
The selected nonbinary constant-weight envelope is
\[
 E_q(\delta,\omega)
 =\min\{U_q^*(\delta),B_q^{\mathrm{pure}}(\delta,\omega),
 B_q^{\mathrm{cc}+\Pp}(\delta,\omega)\}.
\]
In the binary case, the plotted transported envelope is
\[
 E_2(\delta,\omega) =\min\{U_2^*(\delta),
B_2^{\mathrm{MS-CW}}(\delta, \omega)\}.
\]

\subsection{Weight-only and composition-sensitive inputs}
\label{sec:cc-envelope}

Recall that, for a weight-only input $B_q(\delta,\theta)$, understood as a
composition bound by the convention above, one has
\[
 (\Pp B_q)(\delta;p_\omega)
 =\inf_{0\leq\theta\leq1}
 \{K_q(p_\omega,p_\theta)+B_q(\delta,\theta)\}.
\]
Indeed, if $\theta=1-r_0$, then $p_\theta\preceq r$, so the
join formula gives
$K_q(p_\omega,p_\theta)\leq K_q(p_\omega,r)$.
The reverse comparison follows by choosing $r=p_\theta$.
Consequently, arbitrary target compositions do not improve the transport
of the same weight-only input. 

\paragraph{Target compositions in smaller alphabets.} Let $q \geq 3$, $(q-2)/(q-1) \leq \omega \leq 1$.
Recall that $p_1 = (0,1/(q-1),\dots,1/(q-1))$ and define $r_t = (0, 1-t, t/(q-2), \dots, t/(q-2))$. 
Since
$p_\omega \preceq p_1 \preceq r_t$,
\[
 K_q(p_\omega,r_t) = H_q(p_\omega) - H_q(r_t) = K_q(p_\omega,p_1) + K_q(p_1,r_t).
\]
Hence, for a $(q-1)$-ary constant-weight bound $L_{q-1}(\delta,t)$,
measured in base $(q-1)$,
\[
\inf_{0\leq t\leq1}
 \left\{K_q(p_\omega,r_t)
 +\log_q(q-1)L_{q-1}(\delta,t)\right\} =
 h_q(\omega)-\log_q(q-1)
 +\log_q(q-1)(\mathrm{BE}_{q-1}L_{q-1})(\delta).
\]

For $q=3$ this gives
\[
\inf_{0\leq t\leq1} \{K_3(p_\omega,r_t)+\log_3 2\,L_2(\delta,t)\} = K_3(p_\omega,p_1)+\log_3 2
 \inf_{0\leq t\leq1}\{1-h_2(t)+L_2(\delta,t)\}.
\]

The following finite inequality is especially useful near $\omega=1$.
\begin{lemma}[Ternary-to-binary symbol merging]
\label{lem:cc-symbol-merging}
Let $r=(\varepsilon,t,1-\varepsilon-t)$ and $d>2n\varepsilon$.
Then
\[
 A_3(n,d; nr)
 \leq A_2(n,d-2n\varepsilon,nt).
\]
Consequently, for any binary constant-weight rate bound $L_2(D,t)$,
\[
 R_3(\delta;r)\leq \log_3 2\,L_2(\delta-2\varepsilon,t), \qquad\delta>2\varepsilon.
\]
\end{lemma}
\begin{proof}
Map both symbols $0,2$ to binary $0$, and symbol $1$ to binary $1$.
The only lost mismatches are $(0,2)$ and $(2,0)$, at most
$2n\varepsilon$ per pair of words. Thus the image has minimum distance
at least $d-2n\varepsilon>0$; in particular, the map is injective on
the code. Every image has binary weight $nt$. Taking rates proves
the second assertion, with the logarithm-base conversion shown explicitly.
\end{proof}

The merging is not an isometry. The loss $2\varepsilon$ in relative
Hamming distance must not be omitted. More generally, one may retain two
symbols of a $q$-ary alphabet and merge all other symbols, of total
frequency $\varepsilon$, into one of them. The same proof gives a
binary bound with loss at most $2\varepsilon$ and factor $\log_q2$.

\begin{corollary}[Unbalancing without changing the total weight]
\label{cor:cc-same-weight-merging}
Let $q=3$ and $\delta>2(1-\omega)$. For any valid binary
constant-weight bound $L_2$,
\begin{equation}
 R_3(\delta,\omega)
 \leq\log_3 2\inf_{0\leq t\leq\omega/2}
 \left\{\omega\left[1-h_2(t/\omega)\right]
 +L_2(\delta-2(1-\omega),t)\right\}.
 \label{eq:cc-same-weight-bound}
\end{equation}
\end{corollary}
\begin{proof}
Put $a=1-\omega$. Transport $p_\omega=(a,\omega/2,\omega/2)$ to
$r=(a,t,\omega-t)$. An optimal joint distribution of $(X,\Pi)$,
with columns indexed by $\mathrm{id},(12)$, is
\begin{equation}
 \Pr\{X=x,\Pi=\pi\}
 =\frac12\begin{pmatrix}
 a&a\\ t&\omega-t\\ \omega-t&t
 \end{pmatrix}.
 \label{eq:cc-two-permutation-plan}
\end{equation}
Its source marginal is $p_\omega$; its transported marginal is $r$.
Conditionally on either permutation the source entropy is $H_3(r)$,
so the cost is
\[
 H_3(p_\omega)-H_3(r)
 =\omega\log_3 2\,[1-h_2(t/\omega)].
\]
This meets the entropy-potential lower bound and is therefore optimal.
Apply Lemma~\ref{lem:cc-symbol-merging} at the target, and optimize $t$.
Finally use the uniform-nonzero-composition reduction to pass from $p_\omega$ to the full weight shell.
\end{proof}

The target in \eqref{eq:cc-two-permutation-plan} has exactly the original
weight $\omega$. The gain comes from redistributing its nonzero symbols,
not from finding a different balanced weight layer. In particular, the
construction needs only two permutations fixing the zero symbol.

\begin{remark}[The full permutation extension can be essential here]
If $a=1-\omega<t<\omega/2$, the cost in
\eqref{eq:cc-two-permutation-plan} cannot be attained using only cyclic
translations of $\mathbb Z_3$. Indeed, equality in the entropy-potential
inequality for $Y=X+U$ would require $Y$ and $U$ to be independent.
Then, since $r=(a,t,\omega-t)$ has its unique minimum at $0$,
\[
 a=\Pr\{X=0\}=\sum_u\Pr\{U=u\}r_u
\]
forces $U=0$ almost surely. This would give $p_\omega=r$, a
contradiction. Compactness of the additive feasible set makes the cost
gap strict. Thus balanced sources do not make full permutations redundant
when the target composition is unbalanced.
\end{remark}

\subsection{Relation with outer Bassalygo--Elias averaging}
\label{sec:master-BE}

For a composition-dependent upper bound $B_q$, define the Bassalygo--Elias closure operator
\[
 (\BE_q B_q)(\delta) := \inf_{r\in\Delta_{q-1}} \{1-H_q(r)+B_q(\delta;r)\}.
\]
The operator $\BE_q$ yields an unrestricted upper bound; it should be distinguished from the relation-layer Bassalygo--Elias bound~\eqref{eq:relation-eb}.
Indeed, the averaging argument of Lemma~\ref{lem:bassalygo-averaging}, applied to a type class, gives
\[
 A_q(n,d)\leq
 \frac{q^n}{\binom n{\bm w}}A_q(n,d;\bm w).
\]
For a weight-only input, the inheritance convention and
$H_q(r)\leq h_q(1-r_0)$, with equality at $r=p_{1-r_0}$, give
\begin{equation}
 (\BE_q B_q)(\delta)
 =\inf_{0\leq\theta\leq1}
 \{1-h_q(\theta)+B_q(\delta,\theta)\}.
 \label{eq:outer-BE}
\end{equation}
Thus the usual constant-weight conversion is a specialization of the
same operator.

\begin{proposition}[Absorption by outer averaging]
\label{prop:BE-absorption}
Let $u_q=(1/q,\ldots,1/q)$. For every valid input bound $B_q$,
\begin{equation}
 (\Pp B_q)(\delta;u_q)=(\BE_q B_q)(\delta),
 \label{eq:BE-as-uniform-evaluation}
\end{equation}
and
\[
 \BE_q(\Pp B_q)=\BE_q B_q.
\]
\end{proposition}

\begin{proof}
Since $u_q\preceq r$, Theorem~\ref{thm:cc-join-cost} gives
$K_q(u_q,r)=1-H_q(r)$, proving
\eqref{eq:BE-as-uniform-evaluation}. By Theorem~\ref{thm:idempotence},
\[
 (\BE_q(\Pp B_q))(\delta) =(\Pp^2 B_q)(\delta;u_q)
 =(\Pp B_q)(\delta;u_q) =(\BE_q B_q)(\delta),
\]
as desired. 
\end{proof}

\paragraph{Comparison at a fixed composition.}
The join formula also gives
$K_q(p,r)\geq H_q(p)-H_q(r)$, and hence
\begin{equation}
 (\Pp B_q)(\delta;p)
 \geq H_q(p)-1+(\BE_q B_q)(\delta).
 \label{eq:fixed-source-BE-lower}
\end{equation}
Equality holds if an outer minimizing target $r_*$ satisfies
$p\preceq r_*$, by Corollary~\ref{cor:cc-zero-cost}.
For a finite outer infimum, a minimizing sequence $r_j$ suffices
if
\[
 H_q(r_j)-H_q(p^\downarrow\vee r_j^\downarrow)\longrightarrow0.
\]
Moreover, since $u_q=p_{(q-1)/q}$,
\eqref{eq:fixed-source-BE-lower} and
\eqref{eq:BE-as-uniform-evaluation} imply
\[
 \inf_{0\leq\omega\leq1}
 \{1-h_q(\omega)+(\Pp B_q)(\delta;p_\omega)\}
 =(\BE_q B_q)(\delta).
\]
Thus outer averaging over symmetric source compositions gives the same
absorption identity.

In particular, for a weight-only input and a minimizer $\theta_*$
in \eqref{eq:outer-BE},
\begin{equation}
 (\Pp B_q)(\delta;p_\omega)
 =h_q(\omega)-1+(\BE_q B_q)(\delta)
 \label{eq:relative-BE-identity}
\end{equation}
whenever $p_\omega\preceq p_{\theta_*}$. By
Corollary~\ref{cor:explicit-profile}, for $q\geq3$ this condition is
\[
 \theta_*\leq\min\{\omega,(q-1)(1-\omega)\}
 \quad\text{or}\quad
 \theta_*\geq\max\{\omega,(q-1)(1-\omega)\},
 \qquad \theta_*\in[0,1].
\]
The same criterion follows from binary majorization when $q=2$.
At a nonuniform source, \eqref{eq:relative-BE-identity} is an
entropy-shifted comparison, not an identification of transport with
the unrestricted conversion $\BE_q B_q$.

\section{Numerics}
\label{sec:numerics}

\subsection{Comparison of the principal explicit bounds}

As explicit examples, let us consider the cases $q = 3$, $\omega = 0.8$ and $\delta \in \{0.17, 0.175, 0.18\}$ (Table~\ref{tab:q3-rates}) and $q = 4$, $\omega = 0.3$ and $\delta \in \{0.15, 0.2, 0.25\}$ (Table~\ref{tab:q4-rates}). We compare the relevant base bounds with our transport envelope in the last row.

\begin{table}[H]
\centering
\small
\setlength{\tabcolsep}{4.0pt}
\begin{tabular}{lccc}
\toprule
Upper bound & $\delta=0.17$ & $\delta=0.175$ & $\delta=0.18$\\
\midrule
Shortened $q$-ary mixed channel &0.6603202&0.6516413&0.6429629\\
Shortened EB + BHL &0.6522021&0.6436794&0.6351914\\
Permutation transport + cc code-anticode &0.6418900&0.6349414&0.6280642\\
Permutation transport + cc Plotkin &0.6246678&0.6168063&0.6089954\\
Relation-layer Bassalygo--Elias &0.6265477&0.6187202&0.6109420\\
Ben-Haim--Litsyn &0.6260611&0.6171344&0.6082347\\
\textbf{Permutation transport + BHL}
&\textbf{0.6124318}&\textbf{0.6039092}&\textbf{0.5954211}\\
\bottomrule
\end{tabular}
\caption{Upper bounds for $q=3$ and $\omega=0.8$}
\label{tab:q3-rates}
\end{table}

\begin{table}[H]
\centering
\small
\setlength{\tabcolsep}{4.0pt}
\begin{tabular}{lccc}
\toprule
Upper bound & $\delta=0.15$ & $\delta=0.2$ & $\delta=0.25$\\
\midrule
Shortened $q$-ary mixed channel &0.7376107&0.6649032&0.5901030\\
Shortened EB + BHL &0.7303165&0.6556511&0.5830799\\
Permutation transport + cc code-anticode &0.4267979&0.3646439&0.3075474\\
Permutation transport + cc Plotkin &0.4160005&0.3464810&0.2803266\\
Relation-layer Bassalygo--Elias &0.4280715&0.3606050&0.2953136\\
Ben-Haim--Litsyn &0.4146549&0.3375564&0.2628409\\
\textbf{Permutation transport + BHL}&\textbf{0.4087063}&\textbf{0.3340409}&\textbf{0.2614697}\\
\bottomrule
\end{tabular}
\caption{Upper bounds for $q=4$ and $\omega=0.3$}
\label{tab:q4-rates}
\end{table}

In Table~\ref{tab:pt-bhl-certificates} we provide the explicit values of all intermediate and final parameters needed to verify the bounds in the last rows of Tables~\ref{tab:q3-rates} and~\ref{tab:q4-rates}.

\begin{table}[H]
\centering
\small
\setlength{\tabcolsep}{5pt}

\begin{tabular}{lccc}
\toprule
\multicolumn{4}{c}{$q=3$, $\omega=0.8$}\\
\midrule
Parameter & $\delta=0.17$ & $\delta=0.175$ & $\delta=0.18$\\
\midrule
$\text{target weight } \theta$ & 0.1138558 & 0.1185010 & 0.1232558\\
$\nu$ & 0.0000000 & 0.0000000 & 0.0000000\\
$\phi$ & 0.0000000 & 0.0000000 & 0.0000000\\
$\delta_2$ & 0.1700000 & 0.1750000 & 0.1800000\\
$\eta_2$ & 0.0012477 & 0.0013522 & 0.0014636\\
$\xi_2$ & 0.0082550 & 0.0088862 & 0.0095499\\
\bottomrule
\end{tabular}

\vspace{1em}

\begin{tabular}{lccc}
\toprule
\multicolumn{4}{c}{$q=4$, $\omega=0.3$}\\
\midrule
Parameter & $\delta=0.15$ & $\delta=0.2$ & $\delta=0.25$\\
\midrule
$\text{target weight } \theta$ & 0.0922692 & 0.1345841 & 0.1862225\\
$\nu$ & 0.0000000 & 0.0000000 & 0.0000000\\
$\phi$ & 0.0000000 & 0.0000000 & 0.0000000\\
$\delta_2$ & 0.1500000 & 0.2000000 & 0.2500000\\
$\eta_2$ & 0.0008479 & 0.0018618 & 0.0036516\\
$\xi_2$ & 0.0051687 & 0.0107140 & 0.0193281\\
\bottomrule
\end{tabular}

\caption{Certificates for Permutation transport + BHL values in Tables~\ref{tab:q3-rates} and~\ref{tab:q4-rates}}
\label{tab:pt-bhl-certificates}
\end{table}

In all cases shown in Table~\ref{tab:pt-bhl-certificates}, the numerical optimizer satisfies
$\nu=0$ and hence $\phi=0$. Thus the first BHL block has zero mass, and the
two-block construction reduces to the one-block term $R_1$, i.e. to the
Aaltonen constant-weight bound. We therefore omit all parameters associated
with the first block, as well as the intermediate quantities and the separately recorded transport,
BHL and total values. 
Here $\theta$ denotes the target weight in the permutation transport: the source composition $p_\omega$ is
transported to the symmetric composition $p_\theta$, at which the
Ben-Haim--Litsyn bound is evaluated. The remaining $\eta_2$ and $\xi_2$ are the optimizing parameters of the surviving $R_1$ term.

The following figures give the comparisons of the chosen bounds. 
The values are the saved attained numerical certificates.

\begin{figure}[H]
\centering
\includegraphics[width=0.96\linewidth]{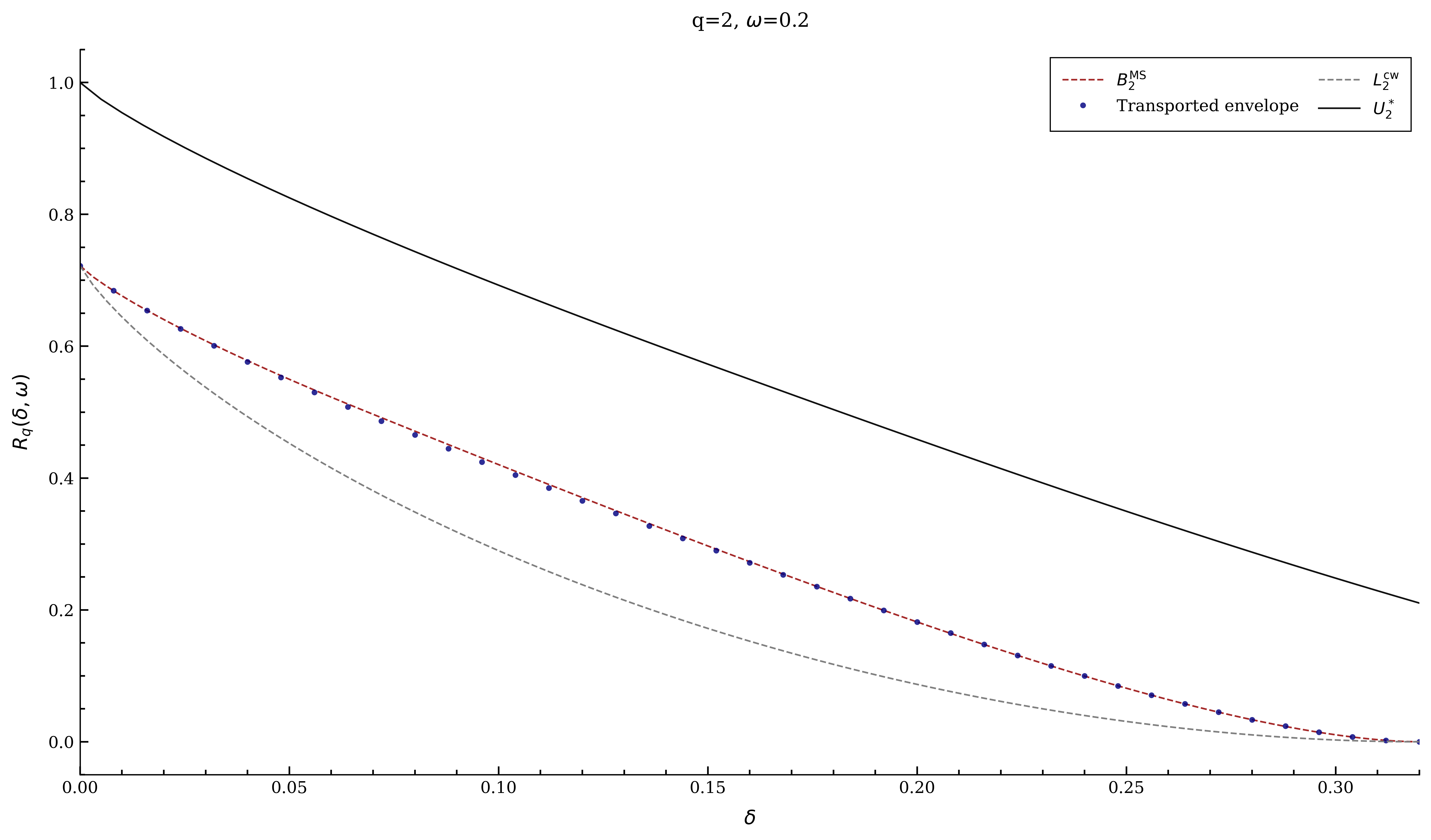}
\caption{Binary constant-weight
comparison for $q=2$ and $\omega=0.2$: the moving-subspace Johnson-layer
input (red dashed), its permutation transport (blue dots), the unrestricted envelope (black solid), and Varshamov--Gilbert lower bound (gray dashed).}
\label{fig:binary-cw-comparison}
\end{figure}

\begin{figure}[H]
\centering
\includegraphics[width=0.96\linewidth]{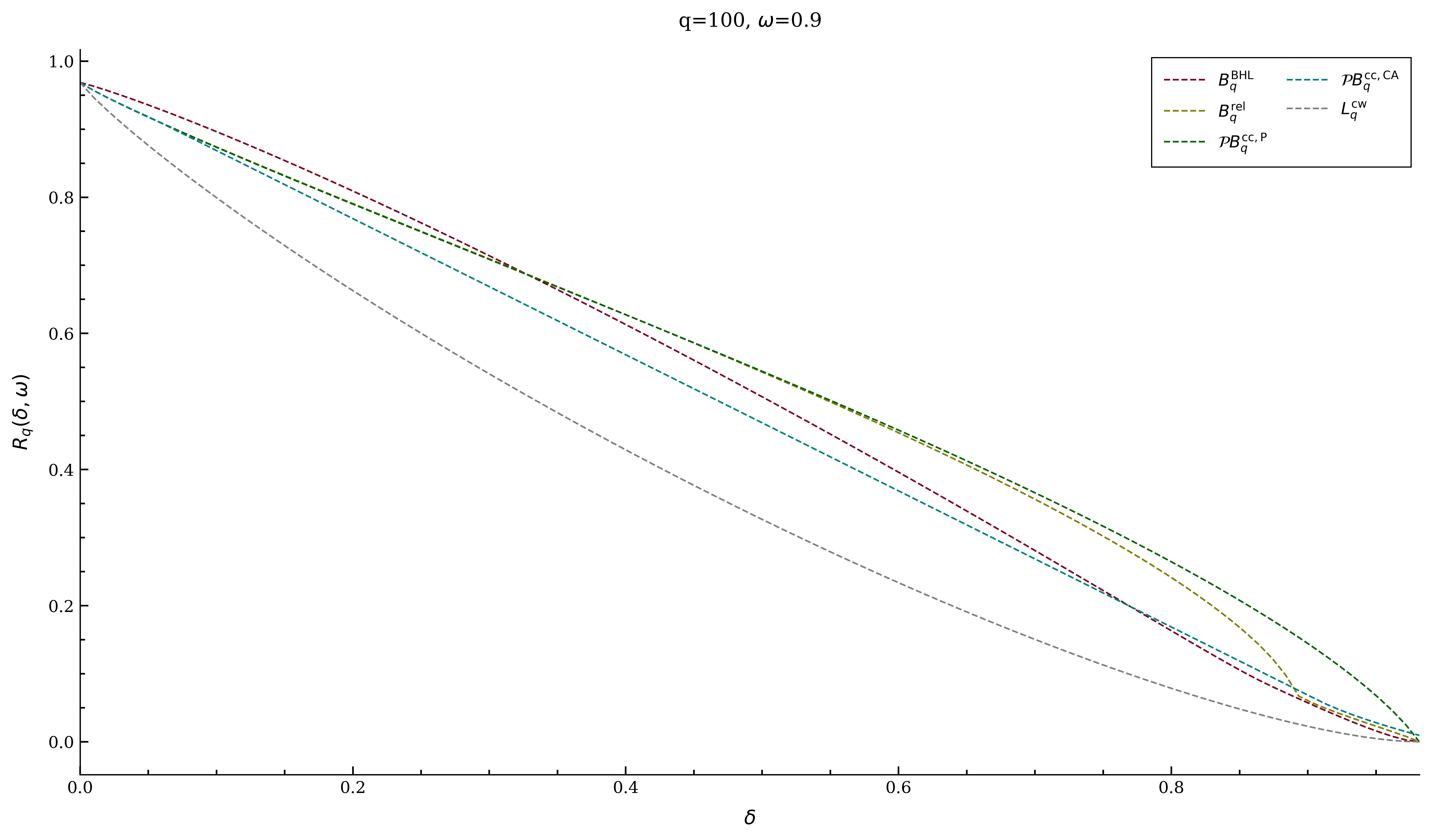}
\caption{Preliminary untransported nonbinary constant-weight bounds for $q=100$ and $\omega=0.9$, together with transported constant-composition bounds:
Ben-Haim--Litsyn (red), Relation-layer Bassalygo--Elias (olive), transported constant-composition Plotkin (green) and code--anticode (blue), and, finally, Varshamov--Gilbert lower bound (gray).}
\label{fig:nonbinary-cw-comparison}
\end{figure}

\begin{figure}[H]
\centering
\includegraphics[width=0.96\linewidth]{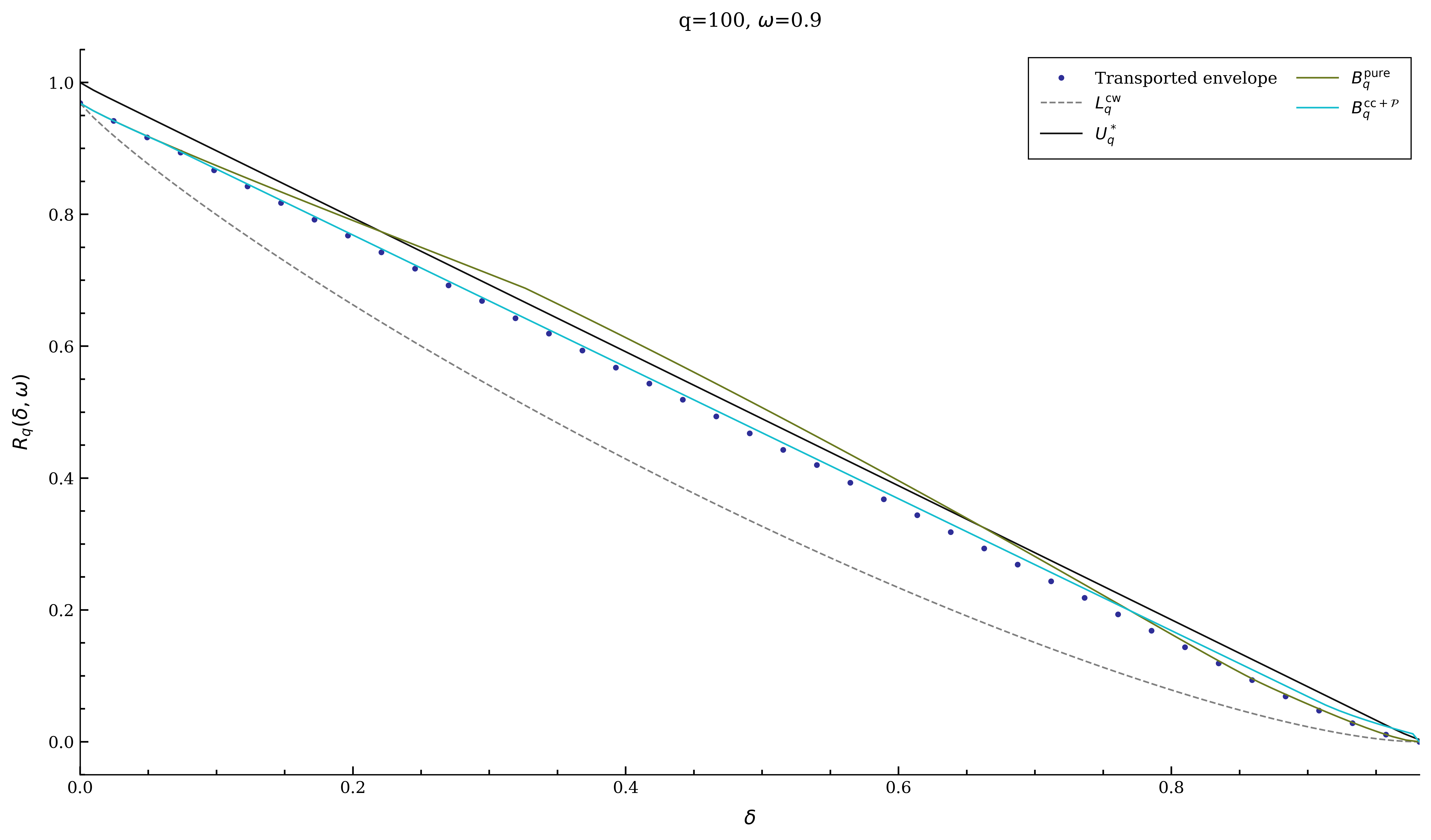}
\caption{Final selected comparison for $q=100$ and
$\omega=0.9$: permutation transport of the selected envelope (blue dots), the pure
constant-weight envelope $\min\{B_q^{\mathrm{rel-BE}}, B_q^{\mathrm{BHL}}\}$ (olive solid), the unrestricted
envelope $U_q^*$ (black solid), the constant-composition envelope $\min\{\Pp B_q^{\mathrm{cc-Pl}}, \Pp B_q^{\mathrm{cc-CA}}\}$ (blue solid) and, finally, Varshamov--Gilbert lower bound (gray dashed).}
\label{fig:nonbinary-final-envelope}
\end{figure}

Permutation transport produces an improvement already in the binary example and this improvement becomes more visible for large $q$.

\subsection{Non-symmetric composition targets}
\label{sec:cc-numerics}

It is natural to ask whether Theorem~\ref{thm:finite-matrix-transport} and Theorem~\ref{thm:finite-permutation-transport} give the same asymptotic results.
By Proposition~\ref{prop:additiveShifts} we know that this is the case when both source and target vectors are symmetric.
However, in the example below the optimal upper bound which can be obtained via Theorem~\ref{thm:finite-permutation-transport}
using the permutation-transport and symbol-merging construction of Corollary~\ref{cor:cc-same-weight-merging} cannot be obtained by Theorem~\ref{thm:finite-matrix-transport}.
This example uses a symmetric source and an asymmetric target.

Take $\omega=0.99$, $\delta=0.17$. The optimal
transport $(0.01,0.495,0.495) \longrightarrow (0.01, t, 0.99-t)$
is the two-permutation plan \eqref{eq:cc-two-permutation-plan}. Its
information cost is $0.99\log_3 2\,[1-h_2(t/0.99)]$.
Merging the rare symbol gives a binary code of relative weight $t$ and relative distance at least $0.15 = \delta - 2(1-\omega)$. Formula
\eqref{eq:cc-same-weight-bound} gives
\[
R_3(0.17,0.99) \leq \log_3 2 \inf_{0\leq t \leq 0.495} \left \{ 0.99 \left( 1-h_2\left(\frac{t}{0.99}\right) \right) + B_2^{\mathrm{MS-CW}} (0.15,t) \right \}.
\]
The numerical value obtained at $t = 0.11576$ is approximately 0.35623.

Transported Ben-Haim--Litsyn bound $\Pp B_3^{\BHL}(0.17; p_{0.99})$ gives 0.36346, while application of Lemma~\ref{lem:cc-symbol-merging} to the constant-weight bound $(\Pp B_2^{\mathrm{MS-CW}})(0.15; p_{0.495})$ gives 0.36134.

\IfFileExists{main.bib}{ 
\bibliography{main}
\bibliographystyle{plain}
}{}

\section*{Appendix A. Connection with rate--distortion}

We now record the rate--distortion interpretation of the profiles obtained
in Section~\ref{sec:weight-transport}.

\begin{proposition}[Exact reduction to Hamming information] \label{pr:rate-distortion-reduction}
Let $X\sim p$, where $p$ is an arbitrary probability vector on $\A$.
For every $\theta\in[0,1]$,
\begin{equation} \label{eq:exact-distortion-profile}
 K_q(p,p_\theta) = \min_{P_{Z \mid X}:\,\Pr\{X\neq Z\}=\theta}I_q(X;Z),
\end{equation}
where $Z$ takes values in $\A$.
\end{proposition}

\begin{proof}
For any feasible permutation transport, put $Z=\Pi^{-1}(0)$.
Then $\Pi(X)=0$ if and only if $X=Z$, so
$\Pr\{X\neq Z\}=\theta$. Since $Z$ is a function of $\Pi$,
data processing inequality~\eqref{eq:dataProcessing} gives $I_q(X;Z)\leq I_q(X;\Pi)$.

Conversely, given a joint distribution of $(X,Z)$ with exact error
probability $\theta$, choose $\Pi$, conditionally on $Z=z$,
uniformly among the $(q-1)!$ permutations satisfying $\Pi(z)=0$,
independently of $X$ given $Z$. The output is zero when $X=Z$
and is uniform on $\A\setminus\{0\}$ otherwise. Hence
$\Pi(X)\sim p_\theta$. Moreover, $Z$ is determined by $\Pi$
and $X\to Z\to\Pi$ is a Markov chain, so
$I_q(X;\Pi)=I_q(X;Z)$. 
\end{proof}

\paragraph{Relation to Erokhin--Pinkston reverse water filling.}
For a discrete source $X\sim p$ and Hamming distortion, the classical
rate--distortion function is
\[
 \mathsf R_p(D) := \min_{P_{Z \mid X}:\,\Pr\{X\neq Z\}\leq D} I_q(X;Z) = 
 \min_{0\leq\theta\leq D}K_q(p,p_\theta),
 \qquad 0\leq D\leq1,
\]
where $Z$ takes values in the same alphabet $\A$, and the second
equality follows from Proposition~\ref{pr:rate-distortion-reduction}.
The explicit solution due to Erokhin and Pinkston is commonly described
as \textit{reverse water filling}~\cite{Erokhin1958,Pinkston1967,Berger1971}: as the allowed distortion
increases, the least probable symbols successively leave the support
of an optimal reproduction distribution. The function $\mathsf R_p(D)$
becomes zero at $D=1-\max_x p_x$ and remains zero thereafter.

In contrast, the transport cost $K_q(p,p_\theta)$ corresponds to the
\textit{exact} error probability
$\Pr\{X\neq Z\}=\theta$, as established in
\eqref{eq:exact-distortion-profile}. It need not be nonincreasing in
$\theta$ and may become positive again after its zero-information
interval. Thus the classical rate--distortion function is obtained by
minimizing the exact-distortion profile over all errors not exceeding
$D$, rather than by evaluating that profile at $\theta=D$.
The join formula in Theorem~\ref{thm:cc-join-cost} provides a single
geometric description of the entire exact-distortion profile.

\end{document}